\documentclass[3p]{elsarticle}

\usepackage{hyperref}
\usepackage{amsmath}
\usepackage{amssymb}
\usepackage{amsthm}
\usepackage{bm} % for bold math
\usepackage{booktabs} % for tables
\usepackage{multirow}
\usepackage{tikz}
\usetikzlibrary{shapes,arrows}
\newcommand{\bbR}{\mathbb{R}}
\newcommand{\bbC}{\mathbb{C}}
\newcommand{\bbH}{\mathbb{H}}
\newcommand{\bbCi}{\mathbb{C}_{\bm i}}
\newcommand{\bbCj}{\mathbb{C}_{\bmj}}
\newcommand{\dt}{\mathrm{d}t}

\newcommand{\bmj}{\bm j}
\newcommand{\bmi}{\bm i}
\newcommand{\bmk}{\bm k}
\newcommand{\bmmu}{\bm \mu}

\renewcommand{\epsilon}{\varepsilon}
\renewcommand{\phi}{\varphi}
\newcommand{\intinf}{\int_{-\infty}^{+\infty}}
\newcommand{\defeq}{\overset{\scriptstyle{\Delta}}{=}}
\newcommand{\atand}{\mathrm{atan2}}
\newcommand{\ip}[1]{\left\langle #1 \right\rangle}

\newcommand{\ie}{\emph{i.e. }}
\newcommand{\sign}{\mathrm{sign}}
\newcommand{\Span}[1]{\mathrm{span}\left\lbrace #1\right\rbrace}
\newcommand{\involj}[1]{{#1}^{\ast\bmj}}
\newcommand{\involmu}[1]{{#1}^{\ast\bmmu}}
\newtheorem{definition}{Definition}
\newtheorem{theorem}{Theorem}
\newtheorem{proposition}{Proposition}

\journal{Applied and Computational Harmonic Analysis}

\begin{document}

\begin{frontmatter}
	\title{Time-frequency analysis of bivariate signals}

\author{Julien Flamant}
\address{Univ. Lille, CNRS, Centrale Lille, UMR 9189 - CRIStAL - Centre de Recherche en Informatique Signal et Automatique de Lille, 59000 Lille, France}
\author{Nicolas Le Bihan}
\address{CNRS/GIPSA-Lab, 11 Rue des math\'{e}matiques, Domaine Universitaire, BP 46, 38402 Saint Martin d'H\`{e}res cedex, France}
\author{Pierre Chainais}
\address{Univ. Lille, CNRS, Centrale Lille, UMR 9189 - CRIStAL - Centre de Recherche en Informatique Signal et Automatique de Lille, 59000 Lille, France}

\begin{abstract}
Many phenomena are described by bivariate signals or bidimensional vectors in applications ranging from radar to EEG, optics and oceanography. The time-frequency analysis of bivariate signals is usually carried out by analyzing two separate quantities, e.g. rotary components. We show that an adequate quaternion Fourier transform permits to build relevant time-frequency representations of bivariate signals that naturally identify geometrical or polarization properties. First, the \emph{quaternion embedding} of bivariate signals is introduced, similar to the usual analytic signal of real signals. Then two fundamental theorems ensure that a quaternion short term Fourier transform and a quaternion continuous wavelet transform are well defined and obey desirable properties such as conservation laws and reconstruction formulas. The resulting spectrograms and scalograms provide meaningful representations of both the time-frequency and geometrical/polarization content of the signal. Moreover the numerical implementation remains simply based on the use of FFT. A toolbox is available for reproducibility. Synthetic and real-world examples illustrate the relevance and efficiency of the proposed approach.

%%
%Bivariate signals appear in many applications such as radar, optics, oceanography, etc. Time-frequency analysis of bivariate signals is often done using the so-called rotary components analysis,  in which case polarization properties are not directly taken into account. Here, we show that using an adequate Quaternion Fourier transform, it is possible to construct time-frequency representations of bivariate signals that naturally identify polarization properties. A bivariate counterpart of the analytic signal is introduced, called the \emph{quaternion embedding}, which allows the identification of instantaneous polarization state and phase for monocomponent signals. Polarization spectrogram and polarization scalogram are also defined, offering new tools towards polarization analysis of bivariate signals. Synthetized and real-world examples illustrate the relevance of the approach.
\end{abstract}

\end{frontmatter}

\section{Introduction}

Bivariate signals are a special type of multivariate time series corresponding to vector motions on the 2D plane or equivalently in ${\mathbb R}^2$. They are specific because their time samples encode the time evolution of vector valued quantities (motion or wavefield direction, velocity, etc). Non-stationary bivariate signals appear in many applications such as oceanography \cite{gonella1972rotary,mooers1973technique}, optics \cite{erkmen2006optical}, radar \cite{Ahrabian2013}, geophysics \cite{samson1983pure} or EEG analysis \cite{sakkalis2011review} to name but a few.

In most of these scientific fields, the physical phenomena (electromagnetic waves, currents, elastic waves, etc) are described by various types of quantities which depend on the frequency over different ranges of frequencies. As frequency components evolve with time, time-frequency representations are necessary to accurately describe the evolution of the recorded signal. 

Over the last 20 years, several works have proposed to develop time-frequency representations of bivariate signals. Most authors have made use of {\em augmented} representations \cite{Schreier2003, schreier2010statistical}, {\em i.e.} the real bivariate signal $[x(t), y(t)] \in \bbR^2$ or the complex version $[f(t), \overline{f(t)}] \in \bbC^2$ where $f(t)=x(t) + \bmi y(t)$. The latter has been used to characterize second order statistics of complex valued processes \cite{Picinbono1997b} as well as higher order statistics \cite{POA96I,POA96II}. In both stationary and non-stationary cases, this approach leads to the extraction of kinematics and polarization properties of the signal \cite{olhede2003polarization,Schreier2008,walden2013rotary,Serroukh2000, Rubin-Delanchy2008}. In the deterministic setting, several bivariate extensions of the Empirical Mode Decomposition (EMD) have been proposed \cite{Tanaka2007,Rilling2007} to decompose bivariate signals into ``simple'' components\footnote{A component is here understood as a possibly time varying monochromatic signal.}. Bivariate instantaneous moments were introduced recently in \cite{lilly2010bivariate} and further extended to the multivariate case in \cite{Lilly2012}. The bivariate instantaneous moment method permits to extract kinematic parameters from a bivariate signal for example. However, as noted by the authors in \cite{lilly2010bivariate}, this method is not applicable when the signal is multicomponent, {\em i.e.} consists of several monochromatic signals. 

Existing methods all share the use of the standard complex Fourier transform (or complex-valued Cram\`er representation in the random case). In this article, we demonstrate that a different \emph{deterministic} approach towards bivariate time-frequency analysis is possible thanks to an alternate definition of the Fourier transform borrowed from geometric algebra. 

In the univariate case, \ie when the signal is real-valued, it is well-known that its (complex) Fourier transform obeys Hermitian symmetry. This feature is fundamental for physical interpretation of the Fourier transform. In signal analysis it leads to the definition of the analytic signal \cite{Gabor1946,Ville1948}, which is the first building block towards more sophisticated time-frequency analysis tools. The analytic signal carries exactly the same information as the original real-valued signal, but with zeros in the negative frequency spectrum. Negative frequencies are redundant since they can be obtained by Hermitian symmmetry from the positive frequencies spectrum. The analytic signal permits to define the \emph{instantaneous amplitude} and the \emph{instantaneous phase} of a real signal: it can be seen as the very first time-frequency analysis tool, at least for amplitude modulation signals.

In the bivariate case or when the signal is complex-valued, its (complex) Fourier transform lacks Hermitian symmetry. Both positive and negative frequencies have to be considered. This observation has motivated a type of analysis often called the \emph{rotary spectrum}. In the simple case where the signal is stationary/periodic (depending from the setting, deterministic or random), the kinematic and polarization properties of the signal at frequency $\omega$ can be extracted from both spectrum values at $\omega$ and $-\omega$ \cite{Rubin-Delanchy2008,Schreier2008}. This approach involves a systematic parallel processing scheme of positive and negative frequencies. 
% Faire une phrase sur representation pour omega > 0 seulement. 
Physicists would prefer a tool that directly yields a representation of the bivariate signal content as a function of positive frequencies only. This is the desirable condition of a direct physical interpretation of the proposed representation, as in a spectrogram for instance.

In this work, we handle bivariate signals as complex-valued signals. We show that they can be efficiently processed using the Quaternion Fourier Transform (QFT), an alternate definition of the Fourier transform. As a first benefit, it restores a special kind of Hermitian symmetry in the quaternion spectral domain. The positive frequencies of the quaternion spectrum carry all the information about the bivariate signal. This permits the definition of  the \emph{quaternion embedding of a complex signal}, the bivariate counterpart of the analytic signal. This quaternion-valued signal is uniquely related to the original complex signal.

An interesting feature of quaternion-valued signals is that geometric information is easily accessible. Indeed, polar forms of quaternions permit the definition of {\em phases} to be interpreted geometrically. Much alike the phase and amplitude defined from an analytic signal, we will introduce the same type of parameters to describe complex-valued (\ie bivariate) signals thanks to their \emph{quaternion embedding}. It will appear that this description receives a direct geometrical interpretation. Even the usual Stokes' parameters used by physicists to describe the polarization state of electromagnetic waves will arise as natural parameters.

Furthermore, the quaternion Fourier transform (QFT) permits to define time-frequency representations for multicomponent bivariate signals in a proper manner. We introduce the Quaternion Short-Term Fourier Transform (Q-STFT) and the Quaternion Continuous Wavelet Transform (Q-CWT) and provide new theorems demonstrating their ability to access geometric/polarization features of bivariate signals. The proposed approach leads to a consistent generalization of classical time-frequency representations to the case of bivariate signals. In practice, known concepts such as ridges extraction  extend nicely to the bivariate case, revealing simultaneously time-frequency and polarization features of bivariate signals at a low computational cost since only FFT are needed still.  %any additional algorithmic effort.

The paper is organized as follows. Section \ref{section:quaternionSec2} first reviews usual quaternion calculus elements and introduces the Quaternion Fourier Transform (QFT) in a general setting. Then we consider the special case of bivariate signals and study this QFT. Section \ref{section:quaternionEmbedding} introduces the quaternion embedding of a complex signal as well as a tailored polar form that recovers both  the frequency content and polarization properties. Section \ref{section:TimeFreqRepresentation} addresses some limitations of the quaternion embedding by introducing the Quaternion Short-Term Fourier Transform and the Quaternion Continuous Wavelet Transform. This section proves two fundamental theorems regarding bivariate time-frequency analysis. Section \ref{section:ExtractionInstantaneousQuantities} provides an asymptotic analysis and defines the ridges of the transforms presented in section \ref{section:TimeFreqRepresentation}. Section \ref{section:ExamplesSection6} illustrates the performances of these new tools on several examples of bivariate signals. A (non-exhaustive) summary of notations and symbols used throughout this paper is provided in Table \ref{table:notations}. Appendices gather technical proofs.

\begin{table}
\centering
\begin{tabular}{ccl}

	\cmidrule[1pt]{2-3}
	&\bf\textsf{Symbol} & \bf\textsf{Description}\\
	\cmidrule{2-3}
	\multirow{5}{*}{\rotatebox[origin=c]{90}{Operators}} &$\mathcal{F}\left\lbrace f \right\rbrace$, $\hat{f}$&Fourier transform of $f$\\
	& $\mathcal{H}\left\lbrace f \right\rbrace$& Hilbert transform of $f$\\
	&$f\ast g$& Convolution product between $f$ and $g$ (non commutative in general)\\
	&$Sf$ &Quaternion Short-Term Fourier transform of $f$\\
	&$Wf$& Quaternion Continuous Wavelet Transform of $f$\\
	\cmidrule{2-3}
	\multirow{4}{*}{\rotatebox[origin=c]{90}{Spaces}}&$\bbC_{\bmmu}$& subfield $\bbR \oplus \bmmu \bbR$ isomorphic to $\bbC$, where $\bmmu$ is a pure unit quaternion.\\
	&$L^p(\bbR, \mathbb{X})$& $L^p$-space of functions taking values in $\mathbb{X}$\\
	&$H^2(\bbR, \mathbb{X})$& Hardy space of square integrable functions taking values in $\mathbb{X}$\\
	&$\mathcal{P}$& Poincar\'{e} half-plane $\left\lbrace (x, y) \in \bbR^2 \middle\vert y \geq 0\right\rbrace$.\\
	\cmidrule[1pt]{2-3}
\end{tabular}
\caption{List of symbols used throughout this work}\label{table:notations}
\end{table}

\section{The Quaternion Fourier Transform}\label{section:quaternionSec2}

\subsection{Quaternion algebra}\label{section:quaternionAlgebra}
Quaternions were first described by Sir W.R. Hamilton in 1843 \cite{hamilton1866elements}. The set of quaternions $\bbH$ is one of the four existing normed division algebras, together with real numbers $\bbR$, complex numbers $\bbC$ and octonions $\mathbb{O}$ (also known as Cayley numbers) \cite{conway2003quaternions}. Quaternions form a four dimensional noncommutative division ring over the real numbers. Any quaternion $q\in \bbH$ can be written in its Cartesian form as
\begin{equation}\label{eq:Cartesianformquaternion}
	q = a + b\bmi + c\bmj + d\bmk,
\end{equation}
where $a, b, c, d \in \bbR$ and $\bmi, \bmj, \bmk$ are roots of $-1$ satisfying
\begin{equation}
	\bmi^2 = \bmj^2 = \bmk^2 = \bmi\bmj\bmk = -1.
\end{equation}
The canonical elements $\bmi, \bmj, \bmk$, together with the identity of $\bbH$ form the quaternion canonical basis given by $\left\lbrace 1, \bmi, \bmj, \bmk\right\rbrace$. Throughout this document, we will use the notation $\mathcal{S}(q) = a \in \bbR$ to define the \emph{scalar part} of the quaternion $q$, and $\mathcal{V}(q) = q - \mathcal{S}(q) \in \Span{\bmi, \bmj, \bmk}$ to denote the \emph{vector part} of $q$. As for complex numbers, we can define the real and imaginary parts of a quaternion $q$ as
\begin{equation}
	\mathfrak{R}(q) = a, \: \mathfrak{I}_{\bmi}(q) = b, \:  \mathfrak{I}_{\bmj}(q) = c, \: \mathfrak{I}_{\bmk}(q) = d.
\end{equation}

A quaternion is called \emph{pure} if its real (or scalar) part is equal to zero, that is $a = 0$, \emph{e.g.} $\bmi, \bmj, \bmk$ are pure quaternions. The quaternion conjugate of $q$ is obtained by changing the sign of its vector part, leading to $\overline{q} = \mathcal{S}(q) - \mathcal{V}(q)$. The modulus of a quaternion $q \in \bbH$ is defined by 
 \begin{equation}
 	\vert q \vert^2 = q\overline{q} = \overline{q}q = a^2+b^2+c^2+d^2.
 \end{equation}
 When $\vert q \vert = 1$, $q$ is called an \emph{unit} quaternion. The set of unit quaternions is homeomorphic to $\mathcal{S}^3 = \left\lbrace x \in \bbR^4 \, \middle\vert \, \Vert x \Vert = 1\right\rbrace$, the 3-dimensional unit sphere in $\bbR^4$. The inverse of a non zero quaternion is defined by
 \begin{equation}
 	q^{-1} = \frac{\overline{q}}{\vert q\vert^2}.
 \end{equation}
It is important to keep in mind that quaternion multiplication is noncommutative, that is in general for $p, q \in \bbH$, one has $pq \neq qp$.

\emph{Involutions} with respect to $\bmi, \bmj, \bmk$ play an important role, and are defined as
\begin{equation}
	\overline{q}^{\bmi} = - \bmi q \bmi, \: \overline{q}^{\bmj} = - \bmj q \bmj, \: \overline{q}^{\bmk} = - \bmk q \bmk.
\end{equation}
The combination of conjugation and involution with respect to an arbitrary pure quaternion $\bmmu$ is denoted by
\begin{equation}
	\involmu{q} \defeq \overline{\left(\overline{q}\right)}^{\bmmu} = \overline{\left(\overline{q}^{\bmmu}\right)},
\end{equation}
and for instance $\involj{( a + b\bmi + c\bmj + d\bmk)} = a + b\bmi - c\bmj + d\bmk$.

The Cartesian form (\ref{eq:Cartesianformquaternion}) is not the only possible representation for a quaternion. Most of the content of this paper is indeed based on other useful representations, which better reflect the geometry of $\bbH$. First, we introduce the Cayley-Dickson form, which decomposes a quaternion $q$ into a pair of 2D numbers isomorphic to a pair of complex numbers. The most general Cayley-Dickson form of a quaternion reads
\begin{equation}\label{eq:CayleyDicksonForm}
	q = q_1 + q_2\bmmu_{\perp}, \: q_1, q_2 \in \bbC_{\bmmu},
\end{equation}
where $\bbC_{\bmmu} \equiv \bbR \oplus \bmmu \bbR$ is a complex subfield of $\bbH$ isomorphic to $\bbC$. Here $\bmmu$ and $\bmmu_\perp$ are pure unit orthogonal quaternions: $\mathcal{S}(\bmmu\bmmu_\perp) = 0$. In particular, the decomposition given by the choice of $\bmmu = \bmi$ and $\bmmu_\perp = \bmj$ in (\ref{eq:CayleyDicksonForm}) will be used extensively.

Just like the polar form of complex numbers, a quaternion Euler polar form can be defined\footnote{The polar form presented here is rather different from what is often referred to as the polar form in the quaternion literature. It usually refers to the decomposition $q = \vert q\vert \exp(\bmmu \theta)$, with $\bmmu$ a pure unit quaternion and $\theta \in \bbR$.}. This form was first introduced by B\"{u}low in \cite{bulow2001hypercomplex}. Any quaternion $q$ has a Euler decomposition which reads
\begin{equation}\label{eq:polarEulerformquaternion}
	q = \vert q\vert\exp(\bmi\theta)\exp(-\bmk \chi)\exp(\bmj \phi),
\end{equation}
where $(\theta, \phi, \chi)$ is called the \emph{phase triplet} of $q$. The phase is (\emph{almost uniquely}) defined within the interval 
\begin{equation}
	(\theta, \phi, \chi) \in [-\pi, \pi[ \times [-\pi/2, \pi/2] \times [-\pi/4, \pi/4].
\end{equation}
The term \emph{almost} refers to the two singular cases for $\chi = \pm \pi/4$, where the phase of a quaternion is not well defined. This phenomenon is known as \emph{gimbal lock}, since the three angles above are indeed the  $xzy$-Euler angles corresponding to the rotation associated to the quaternion $q/\vert q\vert$ of unit norm\footnote{Recall that the unit sphere in $\bbR^4$, $\mathcal{S}^3$ is indeed a two-fold covering of the rotation group $SO(3)$. That is every rotation matrix $\mathbf{R} \in SO(3)$ can be identified with two antipodal points on $\mathcal{S}^3$, $q$ and $-q$. Rotations can be described by three Euler angles \cite{altmann2005rotations}, corresponding to three successive rotations around the canonical axes. The $xzy$-convention is used in this work.}. Figure \ref{fig:polarFormComputation} presents the procedure to obtain the phase triplet $(\theta, \phi, \chi)$. Table \ref{table:quaternionFieldguide} summarizes important properties relevant to quaternion calculus.

\begin{table}
\centering
\begin{tabular}{cl}
	\toprule
	\bf\textsf{Property} & \bf\textsf{Description}\\
	\midrule
	$\mathcal{S}(q), \mathcal{V}(q)$ & Scalar and Vector part of a quaternion $q$\\
	$\mathfrak{R}, \mathfrak{I}_{\bmi}, \mathfrak{I}_{\bmj}, \mathfrak{I}_{\bmk}$& Real and Imaginary parts operators\\
	$\vert q \vert^2 = a^2 + b^2 + c^2 + d^2$& Squared modulus of $q$\\
	$q^{-1} = \overline{q}/\vert q\vert^2$ & Inverse ($q \neq 0$)\\
	$\overline{pq} = \overline{q}\:\overline{p}$& Conjugate of a product\\
	$\overline{q}^{\bmmu} = -\bmmu q \bmmu$ & Involution by unit quaternion $\bmmu$, $\bmmu^2 = -1$\\
	$q = q_1 + q_2\bmj$& Cayley-Dickson form of $q$\\
	$e^{\bmmu \theta} = \cos\theta + \bmmu \sin \theta$&Exponential of a pure unit quaternion $\bmmu$, $\theta \in \bbR$\\
	$q = \vert q\vert e^{\bmi \theta}e^{-\bmk \chi}e^{\bmj \phi}$& Euler polar form of $q$\\
	\bottomrule
\end{tabular}
\caption{Basic guide to quaternion calculus. Here $q = a + b\bmi + c\bmj + d\bmk$ is an arbitrary quaternion, $\bmmu$ is a pure unit quaternion $\bmmu^2 = -1$. The quantities $\theta, \chi, \phi$ are real-valued.}\label{table:quaternionFieldguide}
\end{table}
\begin{figure}
\centering
\begin{tikzpicture}[node distance = 3cm, auto]

\tikzstyle{block} = [rectangle, draw, top color=white, bottom color=black!10, 
    text width=12em, minimum height=3em,, minimum width=3em, font=\footnotesize\sffamily,]
 \tikzstyle{blocklarge} = [rectangle, draw, top color=white, bottom color=black!10, 
    text width=14em, minimum height=3em,, minimum width=3em, font=\footnotesize\sffamily,]
\tikzstyle{textblock} = [rectangle, fill=white, font=\footnotesize\sffamily, text width=8em]
\tikzstyle{textblock2} = [rectangle, fill=white, font=\footnotesize\sffamily]

\tikzstyle{line} = [draw, -latex', line width=.1em]

	\node[textblock2]  at (0, 0) (init) {$\displaystyle q \in \bbH$};
	\node [block] at (3.5, 0) (normalize) {Normalize \\$\displaystyle \tilde{q} = q/\vert q\vert = a + b\bmi + c\bmj + d\bmk$};
	\node [block] at (9, 0) (chicomp) {Compute ellipticity angle\\$\displaystyle \chi = \arcsin[2(bc-ad)]/2$};
	\node [blocklarge] at (12, -2.5) (chi) {If $\chi = \pm \pi/4$\vspace{.5em}\\$\quad \displaystyle \tilde{\theta} = 0$\\$\quad \displaystyle \phi = \atand\left(\frac{2(bd-ac)}{a^2-b^2-c^2+d^2}\right)$\\
	Else\vspace{.5em}\\$\quad \displaystyle \tilde{\theta} = \atand\left(\frac{2(ab+cd)}{a^2-b^2+c^2-d^2}\right)$\\$\quad \displaystyle \phi = \atand\left(\frac{2(bd+ac)}{a^2+b^2-c^2-d^2}\right)$};
	\node [block] at (6, -2.5) (prod) {If $e^{\bmi\theta}e^{-\bmk\chi}e^{\bmj\phi} = -\tilde{q}$ and $\tilde{\theta} \geq 0$\vspace{.5em}\\ $\quad \displaystyle \theta = \tilde{\theta} - \pi$ \vspace{.5em}\\ If $e^{\bmi\theta}e^{-\bmk\chi}e^{\bmj\phi} = -\tilde{q}$ and $\tilde{\theta} < 0$ \vspace{.5em}\\$\quad \displaystyle \theta = \tilde{\theta} + \pi$ \vspace{.5em} \\
	Else $\displaystyle \theta = \tilde{\theta} $};

	\node[textblock] at (1, -2.5) (end) {$q = \vert q \vert e^{\bmi\theta}e^{-\bmk\chi}e^{\bmj\phi}$\vspace{0.5em}\\$\displaystyle \vert q\vert \in \bbR^+$\vspace{0.5em}\\ $\theta \in [-\pi, \pi[$ \vspace{0.5em}\\ $\phi \in [-\pi/2, \pi/2]$\vspace{0.5em}\\ $\chi \in [-\pi/4, \pi/4]$};
	% lines
	\path[line] (init) -- (normalize);
	\draw[line] (normalize) -- (chicomp);
	\path[line] (chicomp.east) -| (chi.north);
	\path[line] (chi.west) -- (prod.east);
	\path[line] (prod.west) -- (end.east);

\end{tikzpicture}
\caption{Computation of the Euler polar form parameters from a quaternion $q$. Adapted from \cite{bulow2001hypercomplex}.}\label{fig:polarFormComputation}
\end{figure}
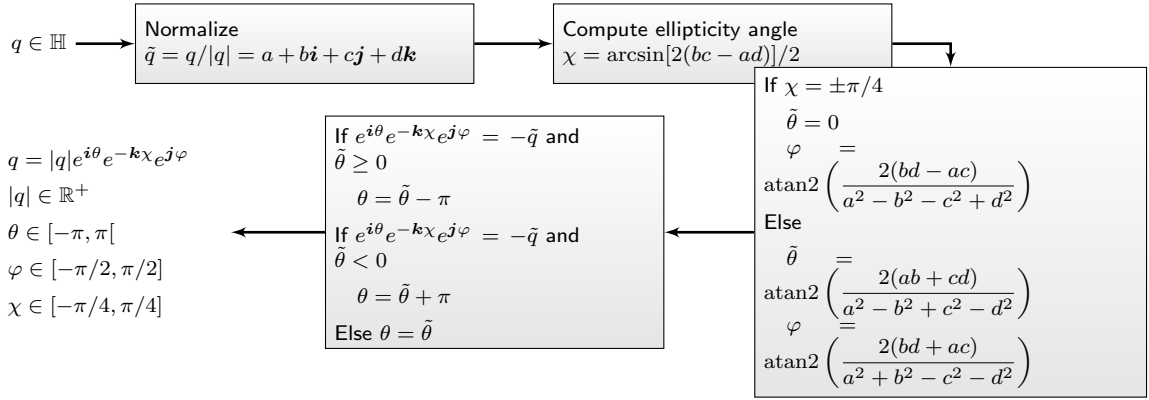

\subsection{Quaternion Fourier Transform}
Most of the literature about \emph{Quaternion Fourier transforms} concerns a particular type of Fourier transforms, namely for functions $f:\bbR^2 \rightarrow \bbH$. For a review on the subject, we refer the reader to \cite{ell2013quaternion} and references therein. This particular set of transforms is for instance of particular interest in color image processing \cite{ell2014quaternion}. In constrast, we review here the one-dimensional Quaternion Fourier transform, for functions $f:\bbR \rightarrow \bbH$.

We define the Quaternion Fourier Transform (QFT) of \emph{axis} $\bmmu$ of a function $f: \bbR \rightarrow \bbH$ by
\begin{equation}\label{eq:definitionQFTaxismu}
	\mathcal{F}\left\lbrace f \right\rbrace =  \intinf f(t)\exp(-\bmmu\omega t)\dt \defeq \hat{f}(\omega).
\end{equation}
In comparison with the standard complex Fourier transform, there are several differences. The position of the exponential function is crucial due to the noncommutative product in $\bbH$, and in our work the exponential function will always be placed on the \emph{right}. Moreover, the axis $\bmmu$ is a free parameter; it is only restricted to be a \emph{pure unit quaternion}. Details on the choice of $\bmmu$ are given in section \ref{section:2choiceofaxis}.

The existence and invertibility of the QFT have been studied by Jamison \cite{jamison1970extension} in his PhD dissertation that dates back to 1970, for functions in $L^1(\bbR, \bbH)$ and $L^2(\bbR, \bbH)$. Although in his definition the exponential function is placed on the left, the proofs are straightforwardly adapted to our convention. We recall the fundamental results only, and refer to Jamison's manuscript for an extensive discussion.

As for the standard complex Fourier Transform, the existency and invertibily of the QFT are first proven for functions in $L^1(\bbR, \bbH)$. The inverse QFT reads
\begin{equation}
	f(t) = \frac{1}{2\pi}\intinf \hat{f}(\omega)\exp(\bmmu\omega t)\mathrm{d}\omega.
\end{equation}
Using a density argument (see \cite[Chapter 2]{mallat2008wavelet} for instance) adapted to the present context, the extension of the QFT to functions in $L^2(\bbR, \bbH)$ can be worked out. The steps are identical to what is known in the complex case and will not be reproduced here.

One can show that $L^2(\bbR, \bbH)$ is a (right) Hilbert space over $\bbH$ \cite{jamison1970extension, teichmuller1936operatoren}. This will be of great interest, since from now on we will write inner products between functions $f, g \in L^2(\bbR, \bbH)$
\begin{equation}
	\ip{f, g} \defeq \intinf f(t)\overline{g(t)}\dt,
\end{equation}
such that the induced norm reads
\begin{equation}
	\Vert f \Vert^2 \defeq \intinf \vert f(t)\vert^2\dt.
\end{equation}
We now prove a fundamental theorem, which extends well-known results to the case of $\bbH$-valued functions. 
\begin{theorem}[Parseval and Plancherel formulas]\label{theorem:ParsevalPlancherel}
Let $f,g \in L^2(\bbR, \bbH)$. Then the following holds
\begin{align}
 \text{(Parseval)}\quad &
	\begin{cases}
	\displaystyle\intinf f(t)\overline{g(t)}\mathrm{d}t &= \displaystyle\frac{1}{2\pi}\intinf \hat{f}(\omega)\overline{\hat{g}(\omega)}\mathrm{d}\omega\\
	\displaystyle\intinf f(t)\involmu{g(t)}\mathrm{d}t &= \displaystyle\frac{1}{2\pi}\intinf \hat{f}(\omega)\involmu{\hat{g}(\omega)}\mathrm{d}\omega
	\end{cases},\\
\text{(Plancherel)} \quad&
\begin{cases}
	\displaystyle\intinf \vert f(t)\vert^2\mathrm{d}t &=\displaystyle\frac{1}{2\pi} \intinf \vert \hat{f} (\omega)\vert^2\mathrm{d}\omega\\
	\displaystyle\intinf f(t)\involmu{f(t)}\mathrm{d}t &= \displaystyle\frac{1}{2\pi}\intinf \hat{f}(\omega)\involmu{\hat{f}(\omega)}\mathrm{d}\omega
	\end{cases}.\label{eq:PlancherelQFT}
\end{align}
\end{theorem}
\begin{proof}
Since Plancherel formulas can be obtained directly with $f=g$ in Parseval formulas, we only give a proof for Parseval's formulas. 
Let $f,g \in L^2(\bbR, \bbH)$. We have
\begin{align}
	\intinf f(t)\involmu{g(t)}\mathrm{d}t &= \frac{1}{2\pi}\int\left(\int \hat{f}(\omega)e^{\bm \mu \omega t}\mathrm{d}\omega\right)\involmu{g(t)}\mathrm{d}t\\
	{(\forall a,b \in \mathbb{H}, \involmu{(ab)} = \involmu{b}\involmu{a})} &= \frac{1}{2\pi}\iint \hat{f}(\omega) \involmu{\left(g(t)e^{-\bm \mu \omega t}\right)}\mathrm{d}t\mathrm{d}\omega \\
	&= \frac{1}{2\pi}\int \hat{f}(\omega) \involmu{\left(\int g(t)e^{-\bm \mu \omega t}\mathrm{d}t\right)} \mathrm{d}\omega \\
	&= \frac{1}{2\pi}\int \hat{f}(\omega) \involmu{\hat{g}(\omega)} \mathrm{d}\omega.
\end{align}
The other equality is proven analogously. 
\end{proof}
Note that this theorem shows two things. It shows that the QFT is an isometry of $L^2(\bbR, \bbH)$. It also shows that another quantity of geometrical nature is preserved by the QFT, the integral $\intinf f(t)\involmu{f(t)}\mathrm{d}t$. This quantity will appear naturally later on in the time-frequency analysis of bivariate signals.

\subsection{General properties}
 We develop some properties related to the Quaternion Fourier Transform (QFT) of axis $\bm\mu$, for arbitrary functions in $L^2(\bbR, \bbH)$. For brevity, most demontrations are omitted, and we refer the reader to \cite{ell2014quaternion} for completeness.

  \paragraph{Linearity} First, let us consider $f, g \in L^2(\bbR, \bbH)$, and denote by $\hat{f}$ and $\hat{g}$ their Fourier transforms. It is straightforward to note that for all $\alpha, \beta \in \bbH$, the QFT of $\alpha f + \beta g$ is $\alpha \hat{f} + \beta \hat{g}$. Therefore the QFT is left-linear. 

  \paragraph{Scaling} Let $\alpha \in \bbR^*$. The QFT of the function $f(t/\alpha)$ is $\vert \alpha \vert \hat{f}(\alpha \omega)$, as can be checked by direct calculation.

  \paragraph{Shifting} Let $\tau \in \bbR$, the QFT of $f(t-\tau)$ is $\hat{f}(\omega)e^{-\bm\mu \omega \tau}$. Let $\omega_0 \in \bbR$, the QFT of $f(t)e^{\bm\mu \omega_0 t}$ is $\hat{f}(\omega-\omega_0)$.
 
  \paragraph{Derivatives} The QFT of the n-th derivative of $f$, $f^{(n)}$ is given by $\hat{f}(\omega)(\bmmu \omega)^n$. Note that multiplication by $(\bmmu \omega)^n$ is from the right, as the exponential kernel was placed on the right in the QFT definition (\ref{eq:definitionQFTaxismu}).
  \paragraph{Invariant subspace} An important feature of the QFT of axis $\bmmu$ defined in (\ref{eq:definitionQFTaxismu}) is that the subspace $L^2(\bbR, \bbC_{\bmmu}) \subset L^2(\bbR, \bbH)$ is an invariant subspace of the QFT of axis $\bmmu$, that is $\mathcal{F}\lbrace{L^2(\bbR, \bbC_{\bmmu})}\rbrace = L^2(\bbR, \bbC_{\bmmu})$. This shows that the restriction of the QFT of axis $\bmmu$ to $L^2(\bbR, \bbC_{\bmmu})$ defines a transform isomorphic to the well-known complex Fourier transform.

  \paragraph{Convolution and product}
  Convolution is one of the cornerstone of signal processing. The following proposition gives the expression for the QFT transform of the convolution product under some conditions. 

  \begin{proposition}[Convolution]\label{proposition:Convolution}
  Let $f \in L^2(\bbR, \bbH)$ and $g \in L^2(\bbR, \bbC_{\bm\mu})$ with respective QFT of axis $\bmmu$ denoted by $\hat{f}$ and $\hat{g}$. Then the QFT of axis $\bmmu$ of the convolution product, $f\ast g$, is given by
  \begin{equation}
    \mathcal{F}\left \lbrace f\ast g \right \rbrace (\omega) = \hat{f}(\omega)\hat{g}(\omega).
  \end{equation}
  Note that in this case the convolution product is not commutative, \ie $f\ast g \neq g \ast f$, as $f$ and $g$ are now quaternion-valued functions.

  \end{proposition}
  \begin{proof}
A direct calculation gives
  \begin{align}
  	 \mathcal{F}\left \lbrace f\ast g \right \rbrace (\omega) &= \int\left( \int f(u)g(t-u)\mathrm{d}u\right)\exp(-\bm\mu \omega t)\mathrm{d}t\\
  	 &= \int f(u)\hat{g}(\omega)\exp(-\bm\mu\omega u)\mathrm{d}u\\
  	 &= \hat{f}(\omega)\hat{g}(\omega)
  \end{align}
  where we have used the fact that $\hat{g}(\omega)$ and $\exp(-\bm\mu\omega u)$ commute, since $\bbC_{\bm\mu}$ is an invariant subspace of the QFT of axis $\bm \mu$.
  \end{proof}

  \begin{proposition}[Product]\label{proposition:Product}
  Let $f \in L^2(\bbR, \bbH)$ and $g \in L^2(\bbR, \bbC_{\bm\mu})$ with respective QFT of axis $\bmmu$ denoted by $\hat{f}$ and $\hat{g}$. Then the QFT of axis $\bmmu$ of their product is given by
  \begin{equation}
    \mathcal{F}\left \lbrace fg \right \rbrace (\omega) =  \frac{1}{2\pi}(\hat{f}\ast\hat{g})(\omega).
  \end{equation}
  \end{proposition}
  \begin{proof}
Similar to the proof of the convolution property.
  \end{proof}
  
  \paragraph{Uncertainty principle}
  Of great importance is the uncertainty principle, also known as Gabor-Heisenberg uncertainty principle. First, considering a function $f \in L^2(\bbR, \bbH)$ we define the temporal mean $u$ like
	\begin{equation}
		u = \frac{1}{\Vert f \Vert^2}\intinf t \vert f(t)\vert^2\mathrm{d}t,
	\end{equation}
	and the mean frequency $\xi$ as
	\begin{equation}
		\xi = \frac{1}{2\pi \Vert f \Vert^2}\intinf \omega \vert \hat{f}(\omega)\vert^2\mathrm{d}\omega.
	\end{equation}
	Spreads around these mean values are defined like
	\begin{equation}
		\sigma_t^2 = \frac{1}{\Vert f \Vert^2}\intinf (t-u)^2 \vert f(t)\vert^2\mathrm{d}t,
	\end{equation}
	\begin{equation}
		\sigma_\omega^2 = \frac{1}{2\pi\Vert f \Vert^2}\intinf (\omega-\xi)^2 \vert \hat{f}(\omega)\vert^2\mathrm{d}\omega.
	\end{equation}

	\begin{theorem}[Gabor-Heisenberg uncertainty principle] Given a function $f \in L^2(\bbR, \bbH)$ with QFT $\hat{f}$ and time (resp. frequency) spread $\sigma^2_t$ (resp. $\sigma^2_\omega$), then the following holds:
	\begin{equation}
		\sigma_t^2\sigma_\omega^2 \geq \frac{1}{4}.
	\end{equation}
	\end{theorem}
	\begin{proof}
	Using a change of variable, it is sufficient to prove the theorem in the case where $u = \xi = 0$. First let us note that
	\begin{equation}
		\sigma_t^2\sigma_\omega^2 = \frac{1}{2\pi \Vert f \Vert^4} \intinf \vert tf(t)\vert^2\mathrm{d}t\intinf \vert \omega \hat{f}(\omega) \vert^2\mathrm{d}\omega
	\end{equation}
	Since $\hat{f}(\omega)\bm j \omega$ is the Fourier transform of $f'(t)$, using the Plancherel identity applied to $\hat{f}(\omega)\bm j \omega$ yields
	\begin{equation}
		\sigma_t^2\sigma_\omega^2 = \frac{1}{\Vert f \Vert^4} \intinf\vert tf(t)\vert^2\mathrm{d}t\intinf \vert f'(t)\vert^2\mathrm{d}t
	\end{equation}
	Schwarz's inequality implies (legit since $L^2(\bbR, \bbH)$ is a Hilbert space)
	\begin{align}
		\sigma_t^2\sigma_\omega^2 & \geq \frac{1}{\Vert f \Vert^4} \left\vert \intinf f'(t) \overline{tf(t)}\mathrm{d}t\right\vert^2\\
		& \geq \frac{1}{\Vert f \Vert^4} \left[\intinf \mathcal{S}\left(f'(t) \overline{tf(t)}\right)\mathrm{d}t\right]^2\\
		&\geq \frac{1}{\Vert f \Vert^4} \left[\intinf  \frac{t}{2}\left( f'(t)\overline{f(t)} + f(t)\overline{f'(t)}\right)\mathrm{d}t\right]^2\\
		&\geq \frac{1}{4\Vert f \Vert^4} \left[\intinf  t(\vert f(t) \vert^2)'\mathrm{d}t\right]^2\\
	\end{align}
	Now, using integration by parts we obtain
	\begin{equation}
	\sigma_t^2\sigma_\omega^2 \geq \frac{1}{4\Vert f \Vert^4} \left[\intinf \vert f(t) \vert^2\mathrm{d}t\right]^2 = \frac{1}{4}.
	\end{equation}
	\end{proof}

\subsection{Bivariate signals and choice of the axis of the transform}\label{section:2choiceofaxis}

%% EXPLICITING bivariate = complex = Ci-valued.
It is well known that bivariate signals can be equivalently described as a pair of real signals or a single complex valued signal. For instance, given a bivariate signal $f(t)$,
\begin{equation}\label{eq:bivariateCivaluedComplexDecomposition}
	f(t) = f_r(t) + \bmi f_{\bmi}(t)
\end{equation}
where $f_r(t)$ and $f_{\bmi}(t)$ are real-valued is a valid decomposition for $f(t)$ in $\bbCi$. Actually, any subfield $\bbC_{\bmmu}$ could have been used to decompose $f(t)$ as they are all isomorphic to the complex field. In the sequel, we will thus use the term bivariate, complex or $\bbCi$-valued interchangeably depending on the context, and will assume that a bivariate signal $f$ is of the form (\ref{eq:bivariateCivaluedComplexDecomposition}).

%% Choice of the axis: back on the CFT (\bmmu = i) and derivation of the choice \bmmu = j (for symmetry reasons.)

There are several possibilities for the choice of the axis $\bmmu$ of the QFT. If one chooses $\bmmu=\bmi$, then this is simply the classical complex Fourier transform. This shows that the QFT definition encompasses the standard complex case, but that it allows different choices as well. This point is interesting as it is well-known that the complex Fourier transform does not exhibit Hermitian symmetry when $f$ is complex-valued. This prevents from defining standard time-frequency tools (\emph{e.g.} the analytic signal) that rely on the Hermitian symmetry property of the transform.  

It is interesting to wonder whether it is possible to obtain a transform that exhibits an  ``Hermitian-like'' symmetry for a $\bbCi$-valued signal? Proposition \ref{prop:ihermitian} below brings a positive answer. Writing down explicitly the QFT of axis $\bmmu$ for a signal as in (\ref{eq:bivariateCivaluedComplexDecomposition}), we have
\begin{equation}
	\hat{f}(\omega) = \int f_r(t)\exp(-\bmmu \omega t)\dt + \bmi \int f_{\bmi}(t)\exp(-\bmmu \omega t)\dt.
\end{equation}
The key idea is now to treat \emph{separately} the real and imaginary part of $f$ in the QFT. That is, we would like to impose to the QFT to maintain the separation between real and imaginary parts of $f$ in the frequency domain. This constraint is satisfied easily provided that $\bmmu$ is orthogonal to $\bmi$, that is $\mathcal{S}(\bmi\bmmu) = 0$. The simplest choice for $\bmmu$ is to take the second imaginary axis in $\bbH$, namely $\bmmu=\bmj$, and we will stick to this choice in the remaining of the paper. Therefore we choose to use the following QFT definition:
\begin{equation}
	\hat{f}(\omega) = \intinf f(t)\exp(-\bmj\omega t)\dt.
\end{equation}
This QFT of a complex-valued signal exhibits a particular symmetry.
\begin{proposition}[$\bmi$-Hermitian symmetry]\label{prop:ihermitian}
	Let $f:\bbR\rightarrow \bbCi$, its QFT satisfies
	\begin{equation}\label{eq:iHermitianSymmetry}
		\hat{f}(-\omega) = -\bmi\hat{f}(\omega)\bmi = \overline{\hat{f}(\omega)}^{\bmi}.
	\end{equation}
\end{proposition}
\begin{proof}
We write $f(t) = f_r(t) + \bmi f_{\bmi}(t)$, with $f_r(t)$ and $f_{\bmi}(t)$ real-valued functions. Therefore, the QFT of axis $\bmj$ of $f$ gives
\begin{equation}
	\hat{f}(\omega) = \hat{f}_r(\omega) + \bmi \hat{f}_{\bmi}(\omega),
\end{equation}
with $\hat{f}_r$, $\hat{f}_{\bmi}$ the QFTs of $f_r, f_{\bmi}$ respectively. Recall that $f_r$ and $f_{\bmi}$ are real-valued functions so that $\hat{f}_r(\omega)$, $\hat{f}_{\bmi}(\omega)$ are $\bbCj$-valued. Therefore, their QFT satisfy the usual Hermitian symmetry (as the QFT for real signals is isomorphic to the standard Fourier transform). As a result:
\begin{equation}
	\hat{f}(-\omega) = \hat{f}_r(-\omega) + \bmi \hat{f}_{\bmi}(-\omega) = \overline{\hat{f}_r(\omega)} + \bmi\overline{\hat{f}_{\bmi}(\omega)} = -\bmi \hat{f}(\omega)\bmi.
\end{equation}
The last equation is obtained recalling that if $z \in \bbCj$, then $\bmi z = \overline{z}\bmi$.
	% OLD PROOF
	% Let us write explicitly $\hat{f}(\omega)$ and $\hat{f}(-\omega)$,
	% \begin{align}
	% \hat{f}(\omega) &= \int f_r(t)\cos(\omega t)\dt+ \bmi \int f_{\bmi}(t)\cos(\omega t)\dt-\bmj\int f_r(t)\sin(\omega t)\dt - \bmk\int f_{\bmi}(t)\sin(\omega t)\dt\\
	% \hat{f}(-\omega)&= \int f_r(t)\cos(\omega t)\dt+ \bmi \int f_{\bmi}(t)\cos(\omega t)\dt+\bmj\int f_r(t)\sin(\omega t)\dt +\bmk\int f_{\bmi}(t)\sin(\omega t)\dt\\
	% \end{align}
	% where $f(t) = f_r(t) + \bmi f_{\bmi}(t)$, with $f_r, f_{\bmi}$ real-valued functions. From the two equations above, one get easily that $\hat{f}(-\omega) = -\bmi\hat{f}(\omega)\bmi$.
\end{proof}

This result is fundamental and will permit to construct the quaternion embedding of a complex signal in section \ref{section:quaternionEmbedding}. This result is, in regard of the Hermitian symmetry of the Fourier transform of real signals, the very first building block of nonstationary bivariate signal processing tools.

\section{Quaternion embedding of complex-valued signals}
\label{section:quaternionEmbedding}

For simple real-valued signals, it is natural to write $f(t) = a(t)\cos[\phi(t)]$, where $a(t) \geq 0$ and $\phi(t)$ are respectively identified as the \emph{instantaneous amplitude} and \emph{phase} \cite{cohen1995time,flandrin1998time}. When the signal is richer, time-frequency analysis techniques aim to provide a set of pairs $[a_k(t), \phi_k(t)]$ that faithfully decribes the signal. Nevertheless, the decomposition $f(t) = a(t)\cos[\phi(t)]$ is the simplest time-frequency model. It has many shortcomings, but understanding it is helpful in understanding how things work.

The analytic signal of $f$ provides a well known way to associate to a simple real signal $f$ an instantaneous amplitude $a(t)$ and phase $\phi(t)$ \cite{cohen1995time}. It is obtained by suppressing negative frequencies from the spectrum \cite{Picinbono1997,Boashash1992}. This operation is motivated by the Hermitian symmetry of real signals spectrum: suppressing negative frequencies spectrum carry no information. This operation associates an unique \emph{canonical} pair $[a(t), \phi(t)]$ to the real signal $f$ \cite{Picinbono1997}. 

When the signal $f$ takes complex values (\ie is bivariate) the analytic signal approach is not directly applicable, since its spectrum no longer exhibits Hermitian symmetry. It was proposed in \cite{lilly2010bivariate} to overcome this limitation by considering an analytic and anti-analytic signal, the latter being the analytic signal associated to the negative part of the spectrum. However, this approach has some inherent limitations due to the complex Fourier transform.

The usual complex Fourier transform brings some intrinsic ambiguity between \emph{geometric content} and \emph{frequency content}. For instance, a linear polarized signal $f$ oscillating in the direction $\exp(\bmi\theta)$ at angular frequency $\omega$ can be written as
\begin{equation}
	f(t) = \exp(\bmi\theta)\exp(\bmi\omega t) = \exp(\bmi\omega t + \bmi\theta).
\end{equation}
One can read either $f$ as a linearly polarized signal at pulsation $\omega$ in the direction $\exp(\bmi\theta)$, or as the analytic signal of $g(t) = \cos(\omega t + \theta)$. Therefore the geometric content (the direction of oscillation) is mixed with the frequency information.

The dimensionality of the complex field is actually responsible for not being able to separate directly the geometric from the frequency content. Using the QFT and the four dimensions of the quaternion algebra will permit to account for the geometric and frequency variables separately. Moreover, the $\bmi$-Hermitian symmetry (\ref{eq:iHermitianSymmetry}) of the QFT of complex-valued signals will lead to the definition of a QFT counterpart of the analytic signal: the \emph{quaternion embedding}. We build upon the recent paper by one of the authors \cite{le2014instantaneous} and present developments on the interpretation of the physical quantities offered by this construction.

\subsection{Definition}

\begin{definition}[Quaternion embedding of complex signals]
Let $f: \bbR \rightarrow \bbCi$. Its quaternion embedding $f_+$ is defined as
\begin{equation}
	f_+(t) = f(t) + \mathcal{H}\left\lbrace f\right\rbrace (t)\:\bmj,
\end{equation}
where $\mathcal{H}\lbrace\cdot\rbrace$ denotes the Hilbert transform
\begin{equation}
	\mathcal{H}\left\lbrace f\right\rbrace(t) = \frac{1}{\pi}\mathrm{p.v.}\intinf \frac{f(\tau)}{t-\tau}\mathrm{d}\tau = \frac{1}{\pi}\mathrm{p.v.}\intinf \frac{f_r(\tau)}{t-\tau}\mathrm{d}\tau + \bmi\frac{1}{\pi}\mathrm{p.v.}\intinf \frac{f_{\bmi}(\tau)}{t-\tau}\mathrm{d}\tau
\end{equation}
where $\mathrm{p.v.}$ stands for Cauchy principal value.
\end{definition}

By construction, $f_+$ is $\bbH$-valued. The original signal $f$ is recovered in the $\left\lbrace 1, \bmi\right\rbrace$-components of $f_+$, and the Hilbert transform of $f$ is found in the $\left\lbrace \bmj, \bmk\right\rbrace$ part of $f_+$. Let us compute the QFT of axis $\bmj$ of the quaternion embedding $f_+$: 
\begin{equation}
	\hat{f}_+(\omega) = \hat{f}(\omega) + \mathcal{F}\left\lbrace 	\mathcal{H}\left\lbrace f\right\rbrace\right\rbrace(\omega)\bmj,
\end{equation}
which can be further decomposed since the quaternion Hilbert transform can be separated into two Hilbert transforms on the real and $\bmi$-imaginary parts of $f$ respectively. As a consequence,
\begin{align}
	\mathcal{F}\left\lbrace \mathcal{H}\left\lbrace f_{r}\right\rbrace\right\rbrace(\omega) &= -\mathrm{sign}(\omega)\hat{f}_r(\omega)\:\bmj, \\
	 \mathcal{F}\left\lbrace \mathcal{H} \left\lbrace f_{\bmi}\right\rbrace\right\rbrace(\omega) &= -\mathrm{sign}(\omega)\hat{f}_{\bmi}(\omega)\:\bmj, 
\end{align}
with $\hat{f}_r, \hat{f}_{\bmi}$ being $\bbCj$-valued so that the QFT of the quaternion Hilbert transform of $f$ reads
\begin{equation}
	\mathcal{F}\left\lbrace \mathcal{H}\left\lbrace f\right\rbrace\right\rbrace(\omega) = -\mathrm{sign}(\omega)\left(\hat{f}_r(\omega)+\bmi \hat{f}_{\bmi}(\omega)\right)\bmj = -\mathrm{sign}(\omega)\hat{f}(\omega)\:\bmj,
\end{equation}
so that
\begin{equation}\label{eq:QFTHextension}
	\hat{f}_+(\omega) = (1+\mathrm{sign}(\omega))\hat{f}(\omega) = 2\:U(\omega)\hat{f}(\omega),
\end{equation}
where $U(\omega)$ is the Heaviside unit step function. Equation (\ref{eq:QFTHextension}) shows that: \emph{(i)} the quaternion embedding of any complex-valued signal has a one-sided spectrum \emph{(ii)} $f$ and $f_+$ have the same frequency content in the positive frequencies region $\omega > 0$, up to a factor $2$. These properties are direct continuations of those of the analytic signal of a real signal.

Let us define the real Hardy space $H^2(\bbR, \bbH)$ such as 
\begin{equation}
\label{eq:Hardyspace}
 	H^2(\bbR, \bbH) = \left\lbrace f\in L^2(\bbR, \bbH) \middle\vert \hat{f}(\omega) = 0 \text{ for all } \omega < 0\right\rbrace.
 \end{equation} 
By construction, the quaternion embedding $f_+$ of a complex signal $f \in L^2(\bbR, \bbCi)$ belongs to $f_+ \in H^2(\bbR, \bbH)$. More precisely, the quaternion embedding method establishes a one-to-one mapping between $L^2(\bbR, \bbCi)$ and $H^2(\bbR, \bbH)$, that is between a complex signal and its quaternion embedding.

\subsection{Instantaneous complex amplitude, ellipticity and phase}

To any real signal one can associate the canonical pair $[a(t), \phi(t)]$ using the polar form of its analytic signal. We show here that one can associate a \emph{canonical triplet} to any complex signal using its quaternion embedding, by means of an appropriate factorization.

In \cite{le2014instantaneous} a complex signal is described by a canonical pair $[a(t), \phi(t)]$, with $a(t)$ and $\phi(t)$ taking complex values. The authors use the recently introduced polar Cayley-Dickson form of a quaternion \cite{sangwine2010quaternion}. While $a(t)$ is interpreted as the instantaneous complex amplitude, the instantaneous phase $\phi(t)$ is restricted to be real, as the meaning of a complex instantaneous phase is not clearly interpretable. This restriction prevents from considering generic bivariate signals. What follows circumvents this issue.

% As we shall see here, the fact that one can obtain a complex phase is actually a consequence of the instantaneous state of polarization of the signal $f$. While it could in theory be possible to retrieve this information directly from this complex phase, we propose an alternative way.\footnote{More precisely, since multiplication is noncommutative in $\bbH$ we have $e^ae^b \neq e^{a+b}$. Rather, the explicit formula for $e^{a+b}$ is given by the Baker-Campbell-Hausdorff formula for quaternions, as an infinite sum. \cite{gilmore2008lie}}

The key idea is to decompose any quaternion embedding $f_+$ using the Euler polar form introduced in section \ref{section:quaternionAlgebra}:
\begin{equation}
	f_+(t) = \vert f_+(t)\vert\exp[\bmi\theta(t)]\exp[-\bmk \chi(t)]\exp[\bmj \phi(t)],
\end{equation}
or equivalently,
\begin{equation}\label{eq:hextensionDefpolar}
	f_+(t) = a(t)\exp[-\bmk \chi(t)]\exp[\bmj \phi(t)].
\end{equation}
This decomposition defines the \emph{canonical triplet} $[a(t), \chi(t), \phi(t)] \in \bbCi \times [-\pi/4, \pi/4] \times [-\pi/2, \pi/2]$. The canonical triplet describes the bivariate signal $f(t)$. By construction $f(t)$ is the projection of $f_+(t)$ onto $\bbCi$ so that
\begin{equation}\label{eq:MESmodel}
	f(t) = a(t)\left[\cos\phi(t)\cos\chi(t) +\bmi \sin\phi(t)\sin\chi(t)\right].
\end{equation}
Note the choice of the negative sign in (\ref{eq:hextensionDefpolar}), which leads to a positive second term in (\ref{eq:MESmodel}).

A clear physical interpretation of the canonical triplet $[a(t), \chi(t), \phi(t)]$ is possible under some usual restrictions. We consider bivariate AM-FM signals, which can be seen as a thorough generalization of the AM-FM univariate signals. A (monocomponent) bivariate AM-FM signal is such that the variations of $\phi(t)$ are much rapid than the rate of variation of the other components,
\begin{equation}\label{eq:polarizedAMFMrequirements}
 	\vert \phi'(t)\vert \gg \vert\chi'(t)\vert, \vert a'(t)\vert/\vert a(t)\vert.
 \end{equation} 

The quantity $\phi(t)$ is called the instantaneous phase of $f$. This choice is natural since $\phi(t)$ appears in (\ref{eq:polarEulerformquaternion}) along the same axis as the QFT, \ie $\bmj$ in our case. The instantaneous frequency of $f$  is thus given by $\phi'(t)$. The term $a(t)$ is referred to as the instantaneous complex amplitude. The term $\chi(t)$ is to be intepreted as the instantaneous ellipticity.

Figure \ref{fig:parametersQuaternionEmbedding} depicts the ellipse traced out by the model (\ref{eq:MESmodel}) over time under condition (\ref{eq:polarizedAMFMrequirements}). Over a short enough period of time, the quantities $a(t) = a$ and $\chi(t) = \chi$ can thus be assumed as constants. The dot on the ellipse represents the value of the complex signal $f(t)$ at a fixed arbitrary time. The argument of $a \in \bbCi$, $\theta = \arg a$ gives the orientation of the instantaneous ellipse. Its modulus $\vert a\vert$ acts as a scale factor of this ellipse. The ellipticity angle $\chi$ controls the shape of the ellipse. If $\chi = 0$, the ellipse becomes a line segment, while if $\chi = \pm \pi/4$ one obtains a circle. The sign of $\chi(t)$ controls the direction of rotation in the ellipse. Finally, the angle $\phi(t)$ gives the position of $f(t)$ on the ellipse. Section \ref{section:ExamplesBivariateSignals} explores further this model with some examples.

% paragraphe "tube as envelopes"
The ellipse parameters can evolve over time. The succession of instantaneous ellipses describe a three-dimensional tube (time being considered as an axis), and thus define the bivariate envelope of the signal.

%paragraphe Lilly 

Lilly and Olhede \cite{lilly2010bivariate} have proposed a model similar to (\ref{eq:MESmodel}), called the Modulated Elliptical Signal (MES) model. It originates from oceanographic signal processing applications with a slightly different parametrization. The quaternion embedding method \emph{a posteriori} justifies the relevance of this model.

\subsection{Instantaneous Stokes parameters}

There is a straightforward connection between the canonical triplet and the notion of polarization state in optics. The pair $[\theta(t), \chi(t)]$ represents exactly the spherical coordinates on the Poincar\'{e} sphere that describes the polarization state of the bivariate signal $f$. Moreover, the square magnitude $\vert a(t) \vert^2$ is the radius of this Poincar\'{e} sphere. 

In practice, one often characterizes the polarization state via the associated Stokes parameters \cite[p. 31]{born2000principles}. In optics, there are four Stokes parameters which read (in the case of fully polarized light) 
\begin{align}
	S_0(t) & = \vert a(t) \vert^2\\
	S_1(t) & = \vert a(t) \vert^2 \cos 2\chi (t) \cos 2\theta(t),\\
	S_2(t) & = \vert a(t) \vert^2 \cos 2\chi (t) \sin 2\theta(t), \\
	S_3(t) & = \vert a(t) \vert^2 \sin 2\chi (t).
\end{align}
  Here $S_0(t)$ is simply the instantaneous energy density of the signal. The three Stokes parameters $S_1, S_2, S_3$ decompose $S_0$ in terms of polarization content, as for all $t$, $S_0^2(t) = S_1^2(t) + S_2^2(t) + S_3^2(t)$. Note that the four Stokes parameters are energetic quantities.

  In constrast with what is generally used in physics, the Stokes parameters defined here are time-dependent. In optics textbooks Stokes parameters are commonly defined as time-average values \cite[Eq. (64), p. 554]{born2000principles}: the time-averaging operator acts with respect to the oscillation of the electromagnetic field. Here, it corresponds to averaging only with respect to the instantaneous phase $\phi(t)$. The so-defined \emph{instantaneous Stokes} parameters thus describe evolving polarization properties.

% Polarization change in a signal is known to lead to the accumulation of a particular quantity, called \emph{Pancharatnam phase}, which a special case of \emph{geometric phase}. In short, such effects are due to the fact that the polarization state evolves with time along curved surface, here the Poincar\'{e} sphere. It is interesting to note that the direct access provided to the pair $(\Psi(t), \chi(t))$ would allow in practice to determine numerically such effects.

The quaternion embedding $f_+$ leads naturally to the Stokes parameters of $f$. The first Stokes parameter is simply obtained by $\vert f_+(t)\vert^2 = S_0(t)$. Moreover basic quaternion arithmetic shows that 
\begin{equation}
	f_+(t)\involj{f_+(t)} = S_1(t) + \bmi S_2(t) - \bmk S_3(t),
\end{equation}
so that the three remaining Stokes parameters are readily obtained from the quantity $f_+(t)\involj{f_+(t)}$. This allows to identify quantities of type $f_+(t)\involj{f_+(t)}$ as natural energetic measures of the polarization content of $f$, while $\vert f_+(t) \vert^2$ is simply an energy density. Note that the two quantities $\vert f(t) \vert^2$ and $f(t)\involj{f(t)}$ are also invariants of the QFT, as stated by the Parseval-Plancherel theorem \ref{theorem:ParsevalPlancherel}.

\begin{figure}
\centering
\includegraphics{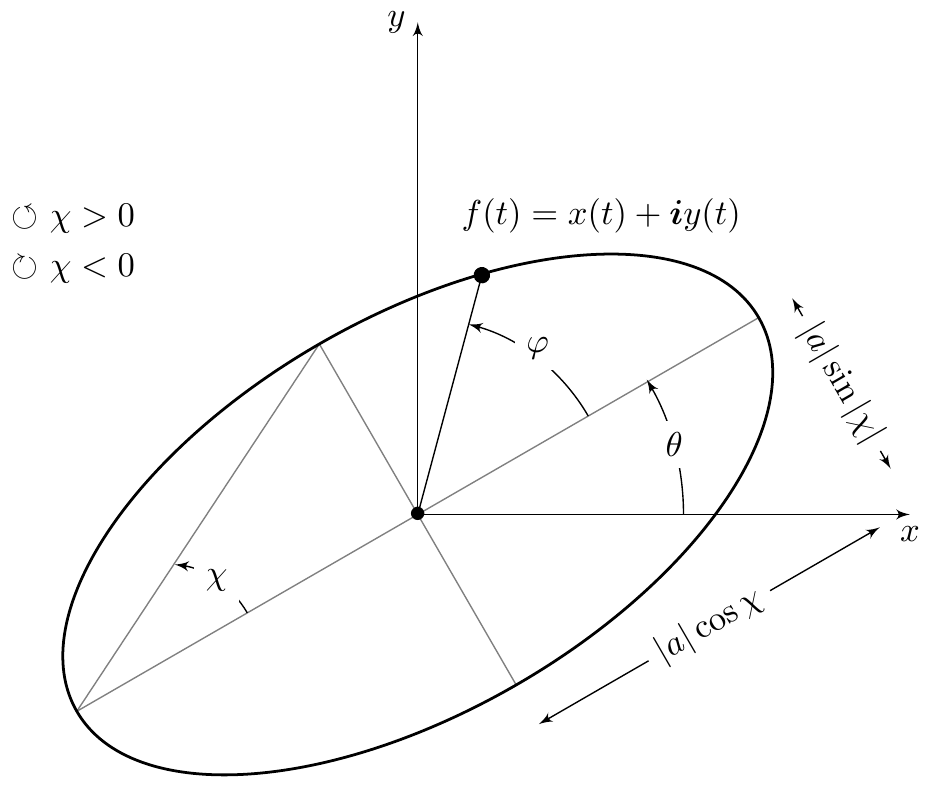}\caption{Representation of the signification of canonical triplet associated to a bivariate signal using its quaternion embedding.}\label{fig:parametersQuaternionEmbedding}
\end{figure}

\subsection{Examples}
\label{section:ExamplesBivariateSignals}
Equation (\ref{eq:MESmodel}) is the bivariate (polarized) counterpart of the well-known AM-FM model $a(t)\cos[\phi(t)]$. Investigating which signals are generated by (\ref{eq:MESmodel}) is thus of great importance. 

The simplest kind of bivariate signals obtained from (\ref{eq:MESmodel}) is the class of \emph{monochromatic polarized signals}. In this case, the instantaneous phase is taken as $\phi(t) = \omega_0 t$, $\omega_0 > 0$, while the other parameters are constant: $a(t) = a_0\exp(\bmi \theta_0) \in \bbCi$ and $\chi(t) = \chi_0 \in [-\pi/4, \pi/4]$. Then
\begin{equation}
	f(t) = a_0\exp(\bmi \theta_0)\left[\cos\chi_0\cos(\omega_0t) + \bmi\sin\chi_0\sin(\omega_0t)\right].
\end{equation}
The angle $\theta_0$ gives the orientation of the signal $f$. Now, browsing through possible values for $\chi_0$, we get
\begin{itemize}
	\item $\chi_0 = 0$: linearly polarized signal, $f(t) = a_0\exp(\bmi \theta_0)\cos(\omega_0t)$.
	\item $\chi_0 \in ]-\pi/4, \pi/4[_{\setminus \left\lbrace 0 \right\rbrace}$: elliptical polarization, counterclockwise (CCW) if $\chi_0 > 0$, clockwise (CW) otherwise. 
	\item $\chi_0 = \pm \pi/4$: circular polarization. It is the degenerate case of the Euler polar form. Then $\theta_0 = 0$, and the signal reads $f(t) = 2^{-1/2}a_0\exp(\pm\bmi \omega_0 t)$. Again, if $\chi_0 = \pi/4$ we have counterclockwise (CCW) polarization and clockwise (CW) if $\chi_0 = -\pi/4$.
\end{itemize}

Figure \ref{fig:chirpExamplesignal}a shows an example of a polarized AM-FM signal specimen. 
Figure \ref{fig:chirpExamplesignal}b show that this signal exhibits a linear chirp in frequency, along with both slowly varying instantaneous orientation and ellipticity. Namely, the instantaneous orientation is flipped slowly, while the instantaneous ellipticity evolves slowly from $\chi > 0$ to $\chi \simeq 0$. 

\begin{figure}
	\centering
	\includegraphics[width=\textwidth]{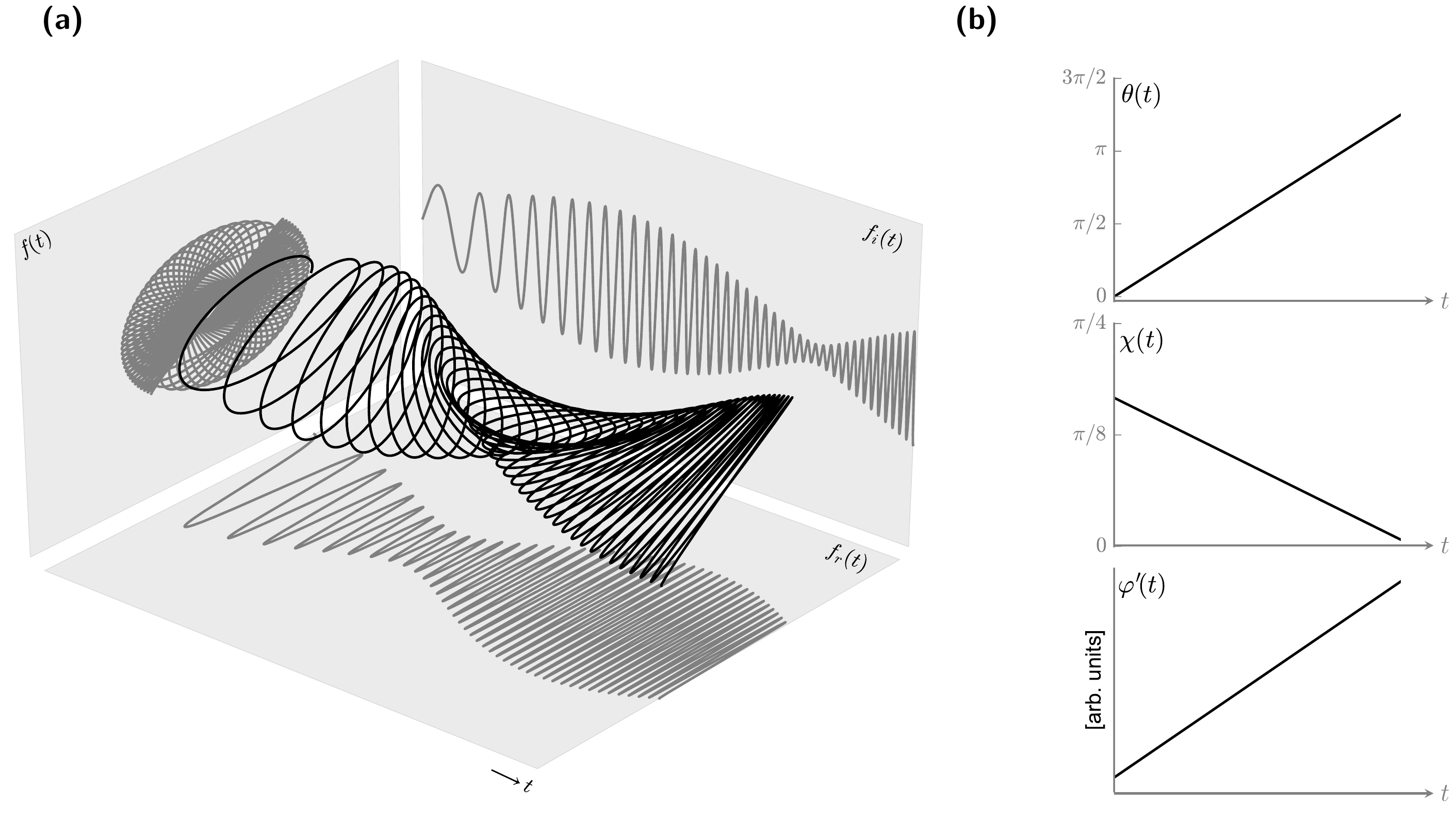}\caption{(a) Example of a signal generated via equation (\ref{eq:MESmodel}). For clarity, the bivariate signal has been developed in time along the third direction. Projection planes show respectively its complex plane trace, as well as real and imaginary components evolving with time. (b) Instantaneous orientation, ellipticity and phase associated with the bivariate signal. Here $\vert a(t) \vert = 1 $.}\label{fig:chirpExamplesignal}
\end{figure}
\subsection{Limitations}
Exactly as in the univariate case, the quaternion embedding does not provide useful information when the signal is multicomponent. Consider the signal $f(t) = \alpha\cos \omega_0t + \alpha \cos \omega_1t$, with $a \in \bbCi$, which is a sum of two linearly polarized signals at angular frequencies $\omega_0$ and $\omega_1$. Its quaternion embedding reads
\begin{equation}
	f_+(t) = 2\alpha e^{\bm j \omega_0 t} + 2\alpha e^{\bm j \omega_1 t} = \alpha \cos\left(\frac{\omega_1-\omega_0}{2}t\right)\exp\left(\bm j\frac{\omega_0+\omega_1}{2}t\right),
\end{equation}
which gives us immediately the Euler polar form, with the canonical parameters given by $\chi(t) = 0$, $a(t) = \alpha \cos\left(\frac{\omega_1-\omega_0}{2}t\right)$ and $\phi(t) = \frac{1}{2}(\omega_0 + \omega_1)t$. While $\chi(t) = 0$ shows that we have indeed linear polarization, the values of $\theta(t)$ and $\phi(t)$ poorly reflect the multicomponent nature of the signal $f$. This motivates the construction of dedicated time-frequency representations for bivariate signals. 

\section{Time-frequency representations of bivariate signals}
\label{section:TimeFreqRepresentation}
We will focus here on two fundamental theorems, leaving the interpretation of such representations to Section \ref{section:ExtractionInstantaneousQuantities}. Examples illustrating the use of the newly introduced time-frequency representations will be given in Section \ref{section:ExamplesSection6}.

Theorem \ref{theorem:CompletenessSTQFT} shows that the quaternion short-term Fourier transform defines a time-frequency representation for bivariate signals, allowing identification of time-frequency Stokes parameters. Theorem \ref{theorem:CompletenessQCWT} shows similar results in the context of time-scale analysis of bivariate signals.

\subsection{Quaternion Short-Term Fourier Transform}
\subsubsection{Definition and completeness}
The quaternion embedding is unable to separate multiple components, so that we introduce the quaternion short-term Fourier transform (Q-STFT). In this section, we suppose that $f \in L^2(\bbR, \bbH)$.

Let $g$ be a real and symmetric normalized window, with $\Vert g \Vert = 1$. For $u, \xi \in \bbR$, its translated-modulated version is
\begin{equation}
	g_{u, \xi}(t) = e^{\bm j \xi t}g(t-u).
\end{equation}
The exponential is on the left. This choice has no influence since $g$ is real, but this is for convenience due to the permutation arising when taking the quaternion conjugate. The functions $g_{u, \xi}(t)$ define \emph{time-frequency-polarization atoms}. The definition of $g_{u, \xi}$ is classical: the term \emph{polarization} solely indicate that the atoms are $\bbCj$-valued, rather than $\bbCi$-valued.

The resulting Q-STFT of $f$ is given by
\begin{equation}
	Sf(u, \xi) = \ip{f, g_{u, \xi}} = \intinf f(t)g(t-u)e^{-\bm j \xi t}\mathrm{d}t.
\end{equation}
This first fundamental theorem to build a time-frequency analysis of bivariate signals ensures energy conservation and reconstruction formula.
\begin{theorem}[Inversion formula and energy conservation]\label{theorem:CompletenessSTQFT}
	Let $f\in L^2(\bbR, \bbH)$. Then the inversion formula reads
	\begin{equation}\label{eq:inversionFormulaSTFT}
		f(t) = \frac{1}{2\pi}\intinf\intinf Sf(u, \xi)g_{u, \xi}(t)\mathrm{d}\xi \mathrm{d}u,
	\end{equation}
	and the energy of $f$ is conserved,
	\begin{equation}\label{eq:energyConservationQSTFT}
		\intinf \vert f(t) \vert^2\mathrm{d}t = \frac{1}{2\pi}\intinf\intinf\vert Sf(u, \xi)\vert^2\mathrm{d}u\mathrm{d}\xi,
	\end{equation}
	as well as the polarization properties of $f$:
	\begin{equation}\label{eq:polarConservationQSTFT}
		\intinf f(t) \involj{f(t)}\mathrm{d}t = \frac{1}{2\pi}\intinf\intinf Sf(u, \xi)\involj{Sf(u, \xi)}\mathrm{d}u\mathrm{d}\xi.
	\end{equation}
\end{theorem}
\begin{proof}
	See \ref{appendix:proofCompletenessSTQFT}.
\end{proof}
This fundamental result extends classical results \cite{mallat2008wavelet} to the bivariate setting. Equation (\ref{eq:energyConservationQSTFT}) shows that $\vert Sf(u, \xi)\vert^2$ defines an energy density in the time-frequency plane. Equation (\ref{eq:polarConservationQSTFT}) shows that the instantaneous state of polarization given by the Stokes parameters is conserved by the representation. Thus we can call the quantity $Sf(u, \xi)\involj{Sf(u, \xi)}$ the \emph{polarization spectrogram} of $f$. 
%By extension, any representation of the time-frequency-polarization properties of $f$ will be called by the same name.

% Note finally that the QSTFT analysis introduced here will now be restricted to the Poincar\'{e} half plane, $\mathcal{P} = \left\lbrace (u, \xi) \in \bbR^2 \middle\vert \xi \geq 0\right\rbrace$ since we will consider only signals in $L^2(\bbR, \bbCi)$. \emph{(This is the same kind of limitation as in the classical case.)}

\subsubsection{Redundancy: the RKHS structure}
The Q-STFT does not span the whole functional space $L^2(\bbR^2, \bbH)$, where $\bbR^2$ stands here for the time-frequency plane. It has a reproducing kernel Hilbert space (RKHS) structure. Writing the Q-STFT of some function $f \in L^2(\bbR, \bbH)$, we have
\begin{equation}
	Sf(u_0, \xi_0) = \intinf f(t)g(t-u_0)e^{-\bm j \xi_0 t}\mathrm{d}t
\end{equation}
for some $(u_0, \xi_0) \in \bbR^2$. The inversion formula (\ref{eq:inversionFormulaSTFT}) yields
\begin{equation}\label{eq:RkhsQSTFT}
	Sf(u_0, \xi_0) = \iint_{\bbR^2} Sf(u, \xi) K(u, \xi, u_0, \xi_0)\mathrm{d}u\mathrm{d}\xi,
\end{equation}
where we have introduced the kernel $K$
\begin{equation}
	K(u, \xi, u_0, \xi_0) = \ip{g_{u, \xi}, g_{u_0, \xi_0}}.
\end{equation}
Equation (\ref{eq:RkhsQSTFT}) shows that the image of $L^2(\bbR, \bbH)$ by the Q-STFT is a RKHS with kernel $K$. Note that this result is a bivariate extension of a property of the usual STFT \cite{mallat2008wavelet}.

\subsubsection{Examples}

To understand the behavior of the Q-STFT, we look at two very simple signals. Section \ref{section:ExtractionInstantaneousQuantities} will carry a more systematic interpretation of these results.

\paragraph{Monochromatic polarized signal} Let $f(t)$ such that its quaternion embedding reads $f_+(t) = a_0\exp(-\bmk \chi_0)\exp(\bmj \omega_0t)$, with $a_0 \in \bbCi$, $\chi_0 \in [-\pi/4, \pi/4]$ and $\omega_0 \in \bbR$. Its Q-STFT reads
\begin{equation}
	Sf_+(u, \xi) = a_0e^{-\bmk\chi_0}\hat{g}(\xi - \omega_0)e^{-\bmj(\xi - \omega_0)u}
\end{equation}
which is localized around the frequency $\xi = \omega_0$ in the time-frequency plane, as expected. The polarization spectrogram of $f_+$ is
\begin{equation}
	Sf_+(u, \xi)\involj{Sf_+(u, \xi)} = \vert \hat{g}(\xi - \omega_0) \vert^2 \left[ \underbrace{\vert a_0 \vert^2\cos2\theta_0\cos2\chi_0}_{S_1} + \bmi \underbrace{\vert a_0 \vert^2\sin2\theta_0\sin2\chi_0}_{S_2} - \bmk \underbrace{\vert a_0 \vert^2 \sin2\chi_0}_{S_3}\right],
\end{equation}
so that the polarization spectrogram of $f_+$ gives immediatly the three time-frequency Stokes parameters that fully characterize $f_+$.
\paragraph{Polarized linear chirp}
Consider a polarized linear chirp $f$, such that its quaternion embedding reads $f_+(t) =  a_0\exp(-\bmk \chi_0)\exp(\bmj \alpha t^2)$, with $a_0 \in \bbCi$, $\chi_0 \in [-\pi/4, \pi/4]$ and $ \alpha \in \bbR$.
When the window is gaussian, that is $g(t) = (\pi \sigma^2)^{-1/4}\exp[-t^2/(2\sigma^2)]$, it is possible to give a closed-form expression for $Sf_+(u, \xi)$ (see e.g. \cite[p. 41]{mallat2008wavelet}):
\begin{equation}
\begin{split}
	Sf_+(u, \xi) =&  a_0e^{-\bmk \chi_0}\left( \frac{2\pi\sigma^2}{1-\bmj2\sigma^2\alpha}\right)^{1/2}\times\\
	&\exp\left[-\frac{(\sigma^2/2)(\xi-2\alpha u)^2}{1+4\sigma^4\alpha^2}\right]\exp\left[\bmj\frac{-(\xi u - \alpha u^2)-\alpha\sigma^4\xi^2}{1+4\sigma^4\alpha^2}\right].
\end{split}
\end{equation}
It shows that $Sf_+(u, \xi)$ is localized in the time-frequency plane around the instantaneous frequency line of $f$, $\xi(u) = 2\alpha u$.
% On this so-called \emph{ridge}, one sees easily from above that the polarization property of $f$ has been conserved.
The polarization spectrogram of $f_+$ is
\begin{equation}
	Sf_+(u, \xi)\involj{Sf_+(u, \xi)} = C(\xi - 2\alpha u)\left[ \underbrace{\vert a_0 \vert^2\cos2\theta_0\cos2\chi_0}_{S_1} + \bmi \underbrace{\vert a_0 \vert^2\sin2\theta_0\sin2\chi_0}_{S_2} - \bmk \underbrace{\vert a_0 \vert^2 \sin2\chi_0}_{S_3}\right],
\end{equation}
where $C$ is a real function depending on the window. One obtains the time-frequency Stokes parameters of the signal.
\subsection{Quaternion continuous wavelet transform}

The fixed size in the time-frequency plane of Q-STFT atoms prevents from analysing a large range of frequencies over short time scales. Following classical theory, we will now introduce the quaternion continuous wavelet transform (Q-CWT). The wavelet atoms will be $\bbCj$-valued, to \emph{mimick} the Q-STFT atoms. The resulting quaternion continuous wavelet transform decouples geometric and frequency content. 

\subsubsection{Definition and completeness}
Throughout this section, we restrict our analysis to the real Hardy space $H^2(\bbR, \bbH)$ introduced in Section \ref{section:quaternionEmbedding}. Recall that one can associate to any $f \in L^2(\bbR, \bbCi)$ a unique element $f_+ \in H^2(\bbR, \bbH)$ called the quaternion embedding of $f$. 

\begin{definition}[Polarization wavelets]\label{definition:polarizationwavelet}
	A polarization wavelet is a $\bbCj$-valued function $\psi \in H^2(\bbR, \bbCj)$, normalized with $\Vert \psi \Vert = 1$ and centered at $t=0$. \emph{Time-scale-polarization atoms} are defined as translated-dilated versions of the wavelet $psi$:
\begin{equation}
	\psi_{u,s}(t) = \frac{1}{\sqrt{s}}\psi\left(\frac{t-u}{s}\right).
\end{equation}
The translation and dilation parameters run over the Poincar\'{e} half plane \ie $(u, s) \in \mathcal{P}$ (see Table \ref{table:notations}). These atoms are normalized, so that $\Vert \psi_{u,s}\Vert = 1$.
\end{definition} 
This definition is again classical. The term \emph{polarization} indicates that the atoms are $\bbCj$-valued, rather than $\bbCi$-valued. 

The quaternion continuous wavelet transform of a bivariate signal $f$ at time $u$ and scale $s$ is:
\begin{equation}
	Wf(u,s) = \ip{f, \psi_{u,s}} = \int_{-\infty}^{+\infty}f(t)\frac{1}{\sqrt{s}}\overline{\psi\left(\frac{t-u}{s}\right)}\dt
\end{equation}

The second fundamental theorem of this paper ensures energy conservation and the existency of an inversion formula for the quaternion continuous wavelet transform.
\begin{theorem}[Inversion formula and energy conservation]\label{theorem:CompletenessQCWT} Let $f \in H^2(\bbR, \bbH)$, and a polarization wavelet $\psi \in H^2(\bbR, \bbCj)$. Suppose that the admissibility condition is satisfied, that is 
\begin{equation}\label{eq:admissibilityCWT}
	C_\psi = \int_0^{+\infty} \frac{\vert\hat{\psi}(\omega)\vert^2}{\omega}\mathrm{d}\omega < \infty.
\end{equation}
The inverse reconstruction formula reads
\begin{equation}\label{eq:inverseReconstructionCWT}
	f(t) = C_\psi^{-1}\int_0^{+\infty}\intinf Wf(u,s)\frac{1}{\sqrt{s}}\psi\left(\frac{t-u}{s}\right)\mathrm{d}u\frac{\mathrm{d}s}{s^2},
\end{equation}
and the energy is conserved,
\begin{equation}\label{eq:energyConservationQCWT}
	\frac{1}{C_\psi}\int_0^{+\infty}\intinf\vert Wf(u,s)\vert^2\mathrm{d}u\frac{\mathrm{d}s}{s^2} = \Vert f\Vert^2,
\end{equation}
as well as the polarization properties
\begin{equation}
	\intinf f(t) \involj{f(t)}\mathrm{d}t = \frac{1}{C_\psi}\int_0^{+\infty}\intinf Wf(u,s)\involj{Wf(u,s)}\mathrm{d}u\frac{\mathrm{d}s}{s^2} \label{eq:polarConservationQCWT}
\end{equation}
\end{theorem} 
\begin{proof}	
See \ref{appendix:proofCompletenessCWT}.
\end{proof}

Equation (\ref{eq:energyConservationQCWT}) indicates that the quantity $\vert Wf(u,s)\vert^2$ can be interpreted as an energy density in the time-scale plane. Equation (\ref{eq:polarConservationQCWT}) means that the polarization properties of $f$ are conserved by the representation. This leads to the definition of the \emph{polarization scalogram}, which is the image of the coefficients $Wf(u,s)\involj{Wf(u,s)}$ in the time-scale plane.

While taking the wavelet $\psi \in H^2(\bbR, \bbCj)$ is necessary to extract time-frequency (or time-scale) tones, the condition that $f\in H^2(\bbR, \bbH)$ is not restrictive. Indeed, since there is a one-to-one correspondence between a bivariate signal $f \in L^2(\bbR, \bbCi)$ and its quaternion embedding $f_+ \in H^2(\bbR, \bbH)$, the results presented here are also valid for signals $f \in L^2(\bbR, \bbCi)$. One then has
\begin{equation}
	Wf(u, s) = \frac{1}{2}Wf_+(u, s).
\end{equation}

\paragraph{Polarization wavelet design} Definition \ref{definition:polarizationwavelet} is classical, and polarization wavelets are constructed similarly to classical analytic wavelets \cite{mallat2008wavelet}.
Namely, they can be built as the frequency modulation of a real, symmetric window $g$, which yields 
\begin{equation}
	\psi(t) = g(t)\exp(\bmj \eta t)
\end{equation}
which admits the QFT $\hat{\psi}(\omega) = \hat{g}(\omega - \eta)$. If $\hat{g}(\omega) = 0$ for $ \vert \omega \vert > \eta$, then the wavelet belongs to $H^2(\bbR, \bbCj)$.

\subsubsection{RKHS structure}
As with the Q-STFT, the image of $L^2(\bbR, \bbH)$ by the Q-CWT does not span the whole $L^2(\mathcal{P}, \bbH)$ space, where $\mathcal{P}$ denotes the time-scale half-upper plane. Rather, it spans only a subspace of it, where the redundancy of the representation is encoded within a RKHS structure. Starting from the definition of the Q-CWT and plugging the inversion formula (\ref{eq:inverseReconstructionCWT}), one gets
\begin{equation}
		Wf(u_0, s_0) = C_\psi^{-1}\int_0^{+\infty}\intinf Wf(u, s) K(u, s, u_0, s_0)\mathrm{d}u\frac{\mathrm{d}s}{s^2},
\end{equation}
where we have introduced the kernel $K$ which reads
\begin{equation}
	K(u, s, u_0, s_0) = \ip{\psi_{u, s}, \psi_{u_0, s_0}}.
\end{equation}
\subsubsection{Example}

Let $f(t)$ be such that its quaternion embedding reads $f_+(t) = a_0\exp(-\bmk \chi_0)\exp(\bmj \omega_0t)$, with $a_0 \in \bbCi$, $\chi_0 \in [-\pi/4, \pi/4]$ and  $\omega_0 \in \bbR$. Its Q-CWT reads
\begin{equation}
	Wf_+(u, s) = a_0e^{-\bmk\chi_0}s^{1/2}\overline{\hat{\psi}(s\omega_0)}e^{\bmj\omega_0u} = a_0e^{-\bmk\chi_0}s^{1/2}\overline{\hat{g}(s\omega_0 - \eta)}e^{\bmj\omega_0u}
\end{equation}
which is localized around scale $s = \eta/\omega_0$ in the time-scale plane. On the line $s = \eta/\omega_0$, one has $Wf_+(u, \eta/\omega_0) = a_0e^{-\bmk\chi_0}\eta^{1/2}\omega_0^{-1/2}\overline{\hat{g}(0)}e^{\bmj\omega_0u}$, so that the polarization information about $f$ can be extracted directly from the Q-CWT coefficients. The polarization scalogram of $f_+$ is
\begin{equation}
	Wf_+(u,s)\involj{Wf_+(u, s)} = s\vert \hat{\psi}(s\omega_0)\vert^2\left[ \underbrace{\vert a_0 \vert^2\cos2\theta_0\cos2\chi_0}_{S_1} + \bmi \underbrace{\vert a_0 \vert^2\sin2\theta_0\sin2\chi_0}_{S_2} - \bmk \underbrace{\vert a_0 \vert^2 \sin2\chi_0}_{S_3}\right],
\end{equation}
leading directly to time-frequency Stokes parameters.
\section{Asymptotic analysis and ridges}\label{section:ExtractionInstantaneousQuantities}

The goal of this section is the so-called \emph{ridge analysis}, that is to provide results about the energy localisation in the time-frequency plane (resp. time-scale plane) of bivariate signals. Ridge analysis has attracted much interest since the 90's. Early work from \cite{delprat1992asymptotic} provided an asymptotic approach. Subsequent theoretical results were developed in a more general setting in \cite{mallat2008wavelet} and in the context of analytic wavelet transform by \cite{lilly2010analytic}.

The discussion here follows closely the approach presented in \cite{delprat1992asymptotic} for univariate signals. It relies upon an \emph{asymptotic hypothesis} on the signal, which essentially means that condition (\ref{eq:polarizedAMFMrequirements}) is satisfied: the phase is varying much faster than the geometric components. Finally, we will discuss how well known algorithms in ridge analysis can be applied to the bivariate setting.

\subsection{Ridges of the quaternion short-term Fourier transform}
The time-frequency-polarization atoms $g_{u, \xi}$ are of the form $g_{u, \xi}(t) = g(t-u)\exp(\bmj\xi t)$, where $g$ is a real, symmetric and normalized window. Under some conditions detailed in \ref{appendix:WindowRidges}, the \emph{ridge} of the transform is given by the set of points $(u, \xi) \in \Omega$ such that
\begin{equation}
	\xi = \xi_{\mathcal{R}}(u) = \phi'(u).
\end{equation}
It gives the instantaneous frequency of the signal. On the ridge, the Q-STFT becomes
\begin{equation}
		Sf(u, \xi_{\mathcal{R}}(u)) \simeq  \sqrt{\frac{\pi}{2}}f_+(u)\frac{g(0)}{\sqrt{\vert \phi''(u)\vert }}e^{\sign{(\phi''(u))}\bmj\frac{\pi}{4}}e^{-\bmj\xi u}.
\end{equation}
This shows that the Q-STFT is on the ridge simply the quaternion embedding of $f$ up to some corrective factor with values in $\bbCj$. As a consequence, assuming that the ridge has been extracted from the Q-STFT coefficients, the polarization properties of $f$ are readily obtained from the polar Euler decomposition of $Sf(u, \xi_{\mathcal{R}}(u))$.

The polarization spectrogram on the ridge is
\begin{equation}
	Sf(u, \xi_{\mathcal{R}}(u))\involj{Sf(u, \xi_{\mathcal{R}}(u))} \simeq \frac{\pi}{2}\frac{g(0)^2}{\vert \phi''(u)\vert}f_+(u)\involj{f_+(u)}.
\end{equation}
On the ridge, one has directly access to the instantaneous Stokes parameters of $f_+$.

% \paragraph{Ridge extraction}
% In a similar way to \cite{delprat1992asymptotic}, the phase of the Q-STFT coefficients permits to locate the ridge more precisely. The only difference with the classical case is that we look at the phase around the $\bmj$-axis, rather than the $\bmi$-axis in the classical setting. 

% Let us introduce $\Theta_S$ the $\bmj$-phase of the QSTFT coefficients
% \begin{equation}
% 	\Theta_S(u, \xi) = \arg_{\bmj}\left[Sf(u, \xi)\right] = \phi(\tau_s) - \xi \tau_s
% \end{equation}
% It follows immediatly that 
% \begin{align}
% 	\frac{\partial \Theta_S(u, \xi)}{\partial u} &= \phi'(\tau_s) - \xi = 0 \text{ at the intersection with the ridge.} 
% \end{align}
% This result gives in particular a very simple way to extract the ridge from the $\bmj$-phase of the Q-STFT.

\subsection{Ridges of the quaternion continuous wavelet transform}

Recall that we consider wavelets $\psi(t) = g(t)\exp(\bmj\eta t)$, where $g$ is a real, symmetric window and $\eta >0$ is the central frequency such that $\psi$ belongs to $H^2(\bbR, \bbCj)$. Under some conditions detailed in \ref{appendix:WindowRidges}, the \emph{ridge} of the transform is given by the set of points $(u, \xi) \in \mathcal{P}$ such that
\begin{equation}
	s = s_{\mathcal{R}}(u) = \frac{\eta}{\phi'(u)},
\end{equation}
which again yields the instantaneous frequency of the signal. The restriction of the Q-CWT to the ridge is:
\begin{equation}
	Wf(u, s_{\mathcal{R}}(u)) \simeq  \sqrt{\frac{\pi}{2s_{\mathcal{R}}(u)}}f_+(u)\frac{g(0)}{\sqrt{\vert \phi''(u)\vert }}e^{\sign(\phi''(u))\bmj\frac{\pi}{4}}.
\end{equation}
The Q-CWT on the ridge is simply the quaternion embedding of $f$ up to some corrective factor with values in $\bbCj$. Computing the Euler polar decomposition of $Wf(u, s_{\mathcal{R}}(u))$ immediatly gives the instantaneous polarization properties. 

The polarization scalogram on the ridge is
\begin{equation}
	Wf(u, s_{\mathcal{R}}(u))\involj{Wf(u, s_{\mathcal{R}}(u))}  \simeq  \frac{\pi}{2s_{\mathcal{R}}(u)}\frac{g(0)^2}{\vert \phi''(u)\vert}f_+(u)\involj{f_+(u)},
\end{equation}
which again shows that the instantaneous Stokes parameters are available on the ridge of the polarization scalogram. 

% \paragraph{Ridge extraction}
% Let us have a closer look at the phase of the Q-CWT coefficients. Let us introduce $\Theta_W$ the $\bmj$-phase of the QCWT coefficients
% \begin{equation}
% 	\Theta_W(u, s) = \arg_{\bmj}\left[Wf(u, s)\right] = \phi(\tau_s) - \phi_\psi\left(\frac{\tau_s - u}{s}\right).
% \end{equation}
% The evaluation of the derivative with respect to $u$ along $\tau_s(u, s) = u_0$ (the so-called wavelet curves in \cite{delprat1992asymptotic}) is:
% \begin{align}
% 	\left.\frac{\partial \Theta_W(u, \xi)}{\partial u}\right\vert_{\tau_s(u, s) = u_0} &= \frac{1}{s}\phi'_\psi\left(\frac{\tau_s - u}{s}\right) + \left.\frac{\partial s}{\partial u}\right\vert_{\tau_s(u, s) = u_0}\frac{\tau_s - u}{s^2}\phi'_\psi\left(\frac{\tau_s - u}{s}\right)\\
% 	&= \frac{\phi'_\psi(0)}{s}\text{ at the intersection with the ridge.} 
% \end{align}
% This result gives a simple way to extract the ridge from the $\bmj$-phase of the Q-CWT coefficients.

\subsection{Ridge extraction and discussion}

It is possible to show that the ridge can be extracted from the $\bmj$-phase of the Q-STFT and Q-CWT coefficients, as originally suggested in the univariate case by \cite{delprat1992asymptotic} (not discussed here). This approach is known to have shortcomings when the signal-to-noise ratio is low, and other approaches have to be used instead \cite{Carmona1997,Carmona1999}. Existing ridge extraction algorithms can be thoroughly adapted to the bivariate setting.

% The asymptotic approach developed above identifies \emph{phase ridge points}. However, it is also possible to estimate local maxima in the time-frequency plane, leading to the identification of \emph{amplitude ridge points}. For instance, it was proposed in \cite{Carmona1997} to extract the ridge using a variational approach based on minimizing an energy functional, which we can express in the Q-STFT setting the following way. The goal is to obtain a function $\gamma$ minimizing the energy functionnal 
% \begin{equation}\label{eq:variationalCarmona}
% 	E_f(\gamma) = -\intinf \vert Sf(u, \gamma(u))\vert^2\mathrm{d}t + \intinf [\lambda \gamma'(u)^2 + \mu \gamma''(u)^2]\mathrm{d}u
% \end{equation}
% where $\lambda$ and $\mu$ are parameters enforcing the smoothness of the function $\gamma$. Existing algorithms minimizing (\ref{eq:variationalCarmona}) such as presented in \cite{Carmona1997,Carmona1999} 

A detailed discussion on ridge extraction methods is out the scope of the present paper. In our simulations we have used a heuristic method which identifies at each instant $u$ the local maxima of the energy density in the time-frequency (resp. time-scale) plane. This method, although not optimal, provides reasonably good results for our purpose.

\section{Time-frequency representations of bivariate signals: illustration}

\label{section:ExamplesSection6}
We finally illustrate the time-frequency representations presented in section \ref{section:TimeFreqRepresentation} and the subsequent results in section \ref{section:ExtractionInstantaneousQuantities}. Two synthetic examples are presented, which are \emph{polarized} counterparts of classical examples (see \emph{e.g.} \cite{mallat2008wavelet}). A real-world example is also provided.

Two equivalent time-frequency representations of bivariate signals can be proposed. The first one is directly related to classical frequency analysis. One can extract ridges from a time-frequency energy density. Instantaneous polarization properties are thus unveiled using the Euler polar form on these ridges. Another approach is to compute the polarization spectrogram (resp. polarization scalogram) of the signal. This gives the three time-frequency Stokes parameters of the signal.

We explore the benefits of the two methods. As they are equivalent representations, we will use the term \emph{polarization spectrogram (resp. scalogram)} to denote one or the other.

% JUlIEN: Mettre ça dans la cover letter? 
% Code reproducing the examples and figures will be soon available at [link]. At the same address, the interested reader will also find a full-working implementation of the tools presented in this paper.

\subsection{Sum of linear polarized chirps}

Consider a superposition of two linear chirps, each having its own polarization properties (\ref{eq:linearchirp1}) and (\ref{eq:linearchirp2}). The signal is defined on the time interval $[0, 1]$ by $N = 1024$ equispaced samples. It can be written as a the superposition $f(t) = f_1(t) + f_2(t)$, where their respective polar Euler decomposition reads
\begin{align}
a_1(t) = \exp\left(\bmi \frac{\pi}{4}\right), \:\chi_1(t) = \frac{\pi}{6} - t, \:\phi_1(t) = 50\pi + 250\pi t\label{eq:linearchirp1}\\
a_2(t) = \exp\left(\bmi\frac{\pi}{4} 10 t\right), \: \chi_2(t) = 0, \:\phi_2(t) = 150\pi + 250\pi t\label{eq:linearchirp2}
\end{align}

This signal can be seen as a polarized version of the classical parallel linear chirps signal \cite{mallat2008wavelet}. The Q-STFT was computed with a Hanning window of size $101$ samples, providing good time-frequency clarity.

Figure \ref{fig:twoLinearChirps} shows the two equivalent polarization spectrograms of $f$. Figure \ref{fig:twoLinearChirps}a, b and c depicts the three time-frequency Stokes parameters. Figure \ref{fig:twoLinearChirps}d, e and f show respectively the time-frequency energy density,  instantaneous orientation and ellipticity.

The three Stokes parameters provide a reading of time-frequency-polarization properties of the two chirps. They have been normalized to be meaningully interpreted. We have normalized $S_1$, $S_2$ and $S_3$ by $S_0$ which is simply the time-frequency energy density depicted in figure \ref{fig:twoLinearChirps}d. While $S_3$ is directly an image of the ellipticity, the orientation has to be recovered by simultaneously inspecting the three Stokes parameters.

The time-frequency energy density permits the identification of the two linear chirps. This time-frequency energy density can be retrieved from the three Stokes parameters, as $S_0^2 = S_1^2 + S_2^2 + S_3^2$. Moreover, \ref{fig:twoLinearChirps}e, \ref{fig:twoLinearChirps}f show instantaneous orientation and ellipticity extracted from the ridge. The polarization properties of each chirp are correctly recovered.

\begin{figure}
	\includegraphics[width=\textwidth]{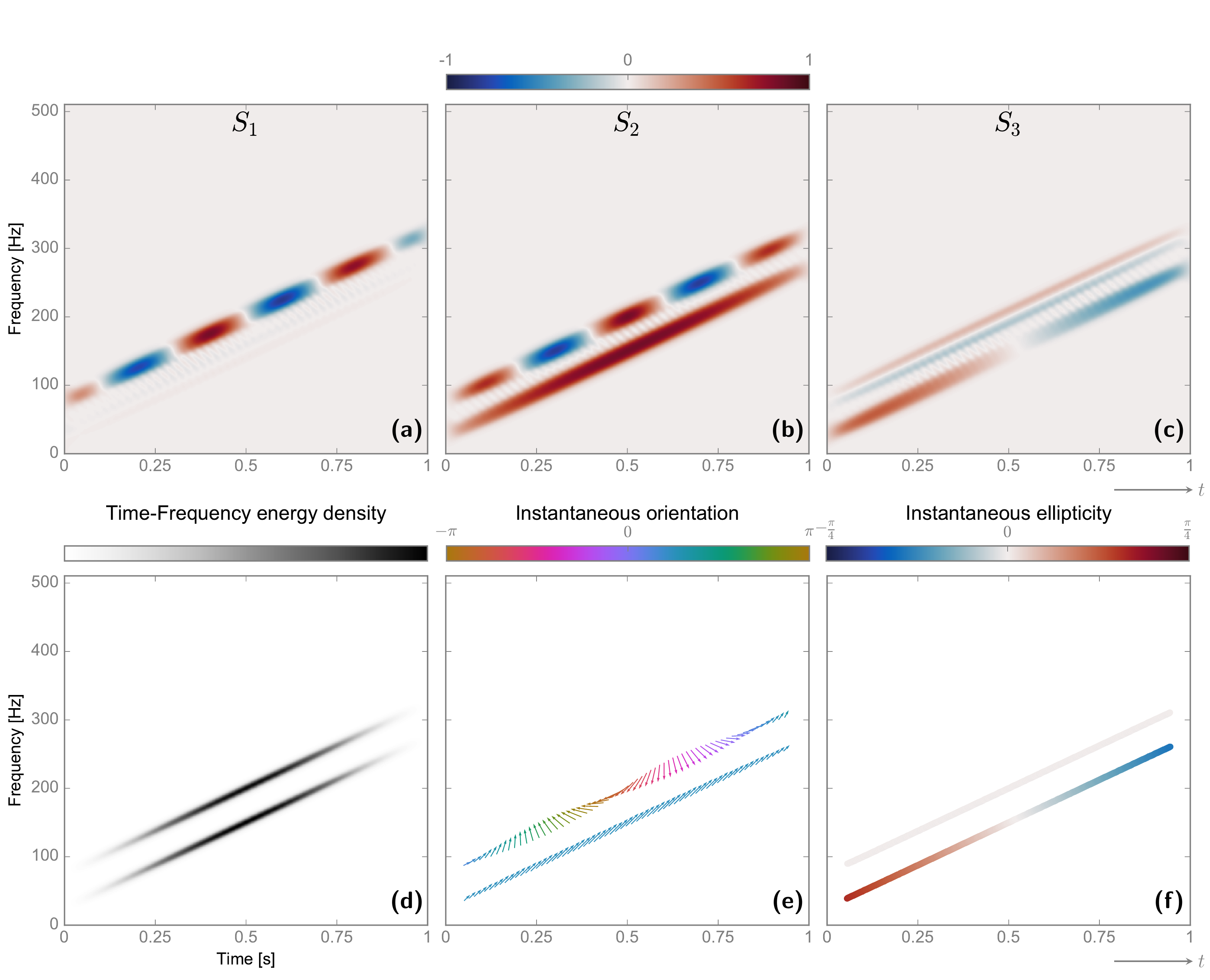}\caption{Sum of two polarized linear chirp and its two polarization spectrogram representation. (a), (b) and (c) Stokes parameters in the time-frequency plane. (d) Time-frequency energy density (e) instantaneous orientation on the ridge (f) instantaneous ellipticity on the ridge. Top chirp is linearly polarized, whereas the bottom chirp shows a slowly varying ellipticity.}\label{fig:twoLinearChirps}
\end{figure}

\subsection{Sum of hyperbolic polarized chirps}

Consider two hyperbolic chirps, each having its own polarization properties. The signal is defined on the time interval $[0, 1]$, with $N = 1024$ samples. It can be written as the superposition $f(t) = f_1(t) + f_2(t)$, where the polar Euler decomposition are 
\begin{align}
a_1(t) = \exp\left(-\bmi \frac{\pi}{3}\right), \:\chi_1(t) = \frac{\pi}{6}, \:\phi_1(t) =\frac{15\pi}{0.8 - t}\\
a_2(t) = 0.8\exp\left(\bmi 5t\right), \: \chi_2(t) = -\frac{\pi}{10}, \:\phi_2(t) = \frac{5\pi}{0.8-t}.
\end{align}
The Q-CWT was computed using a Morlet wavelet with $\eta = 5$.

Figure \ref{fig:twoHyperbolicChirps} shows the two equivalent polarization scalograms of $f$. Figure \ref{fig:twoHyperbolicChirps}a, b, and c represent the three time-scale Stokes parameters $S_1$, $S_2$, $S_3$. Figure \ref{fig:twoHyperbolicChirps}d, e and f give the equivalent representation using the time-scale energy density, and the instantaneous orientation and ellipticity extracted from the ridge. The polarization properties of each chirp are correctly recovered.

These two examples validate the use of the Q-STFT and Q-CWT representations for nonstationary bivariate signals. Time-frequency (resp. time-scale) resolution is governed by the choice of the window (resp. wavelet), so that usual trade-offs apply.

\begin{figure}
	\includegraphics[width=\textwidth]{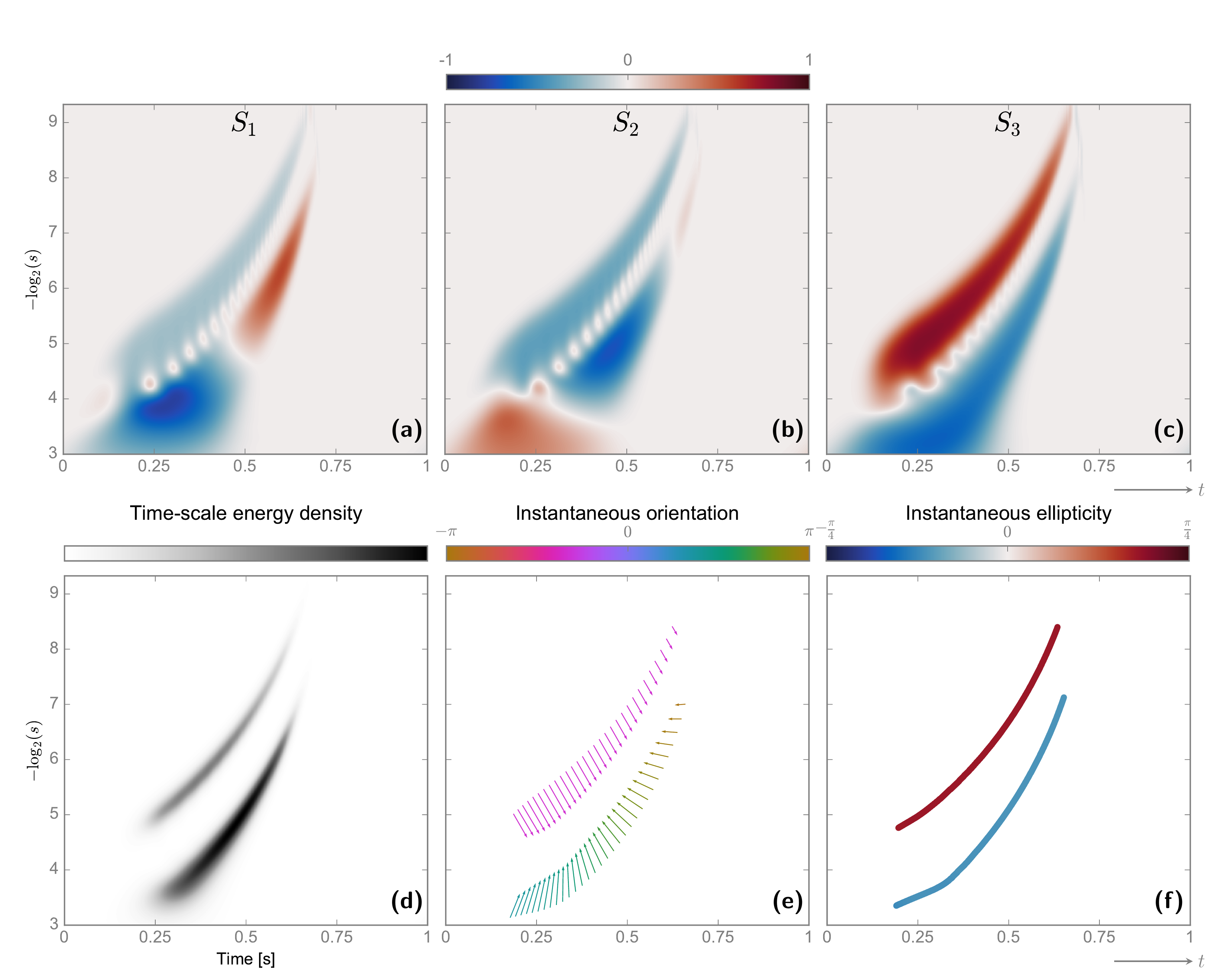}\caption{Sum of two polarized hyperbolic chirps and its polarization scalogram representation. (a), (b) and (c) Stokes parameters in the time-scale plane. (d) Time-scale energy density (e) instantaneous orientation on the ridge (f) instantaneous ellipticity on the ridge. Top chirp shows positive ellipticity, whereas the bottom chirp shows a negative one.}\label{fig:twoHyperbolicChirps}
\end{figure}

\subsection{A real world example}
\begin{figure}
\centering
	\includegraphics[width=0.5\textwidth]{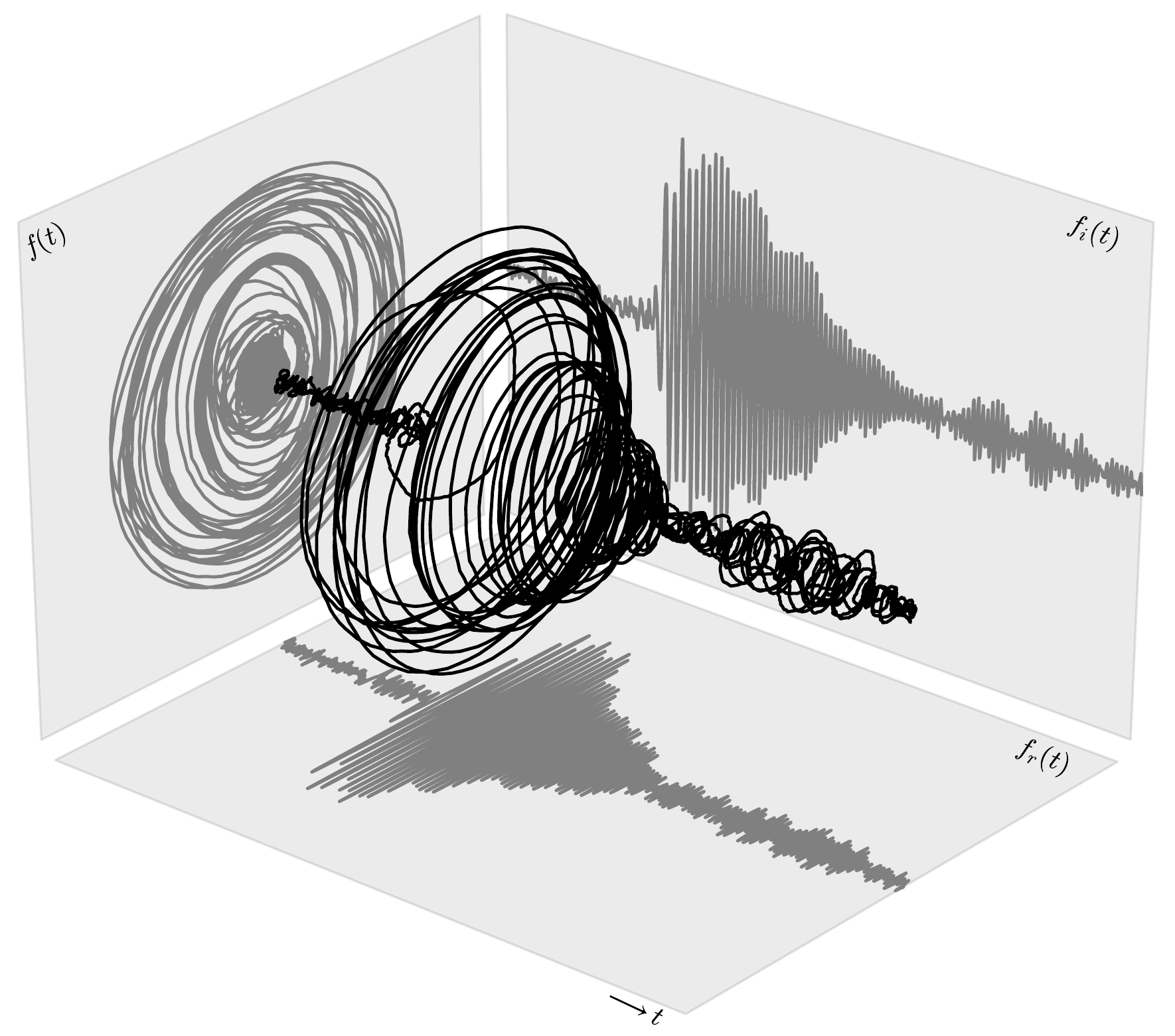}
	\caption{Seismic trace from the 1991 Solomon Islands Earthquake. For clarity, the bivariate signal has been developed in time along the third direction. Projection planes show respectively its complex plane trace, as well as real and imaginary components evolving with time.}\label{fig:SalomonSignal}
\end{figure}

Figure \ref{fig:SalomonSignal} shows a seismic trace of the 1991 Solomon Islands Earthquake. This signal has already been studied by several authors \cite{olhede2003polarization,sykulski2015improper,lilly1995multiwavelet}. Data is available as part of JLab \cite{Lillytoolbox}. It displays the time evolution of the process $f(t) = y(t) + \bmi r(t)$, where $y$ is the vertical component and $r$ is the radial component. The part of the signal which is represented contains $N = 9000$ samples, equispaced by $0.25$ s. . Figure \ref{fig:SalomonSignal} suggests that the signal is on average elliptically polarized, whereas the instantaneous orientation does not appear clearly from the seismic trace. 

The Q-STFT of the signal has been computed using a Hanning window of size $801$ samples, with window spacing equal to $10$ samples. The Q-CWT of the signal has been computed on $200$ scales, and using a Morlet wavelet with $\eta = 5$.

Figure \ref{fig:SalomonBivariateSpectro} and \ref{fig:SalomonBivariateScalo} depict  respectively polarization spectrograms and polarization scalograms of $f$.  In Figure \ref{fig:SalomonBivariateSpectro}d and \ref{fig:SalomonBivariateScalo}d, the ridge has been represented on top of the time-frequency/time-scale energy density. Figure \ref{fig:SalomonBivariateSpectro}e, f and \ref{fig:SalomonBivariateScalo}e, f show the instantaneous orientation and ellipticity on the ridge.

From the ridge, this signal can be described in first approximation as a slow linear chirp in frequency. Moreover, as seen in both descriptions -- Q-STFT and Q-CWT -- the orientation of the signal remains constant in the most energetic part, at around $-100$ degrees. The instantaneous ellipticity is on average equal to $\chi \simeq \pi/5$, confirming the elliptical polarization obtained by visual inspection of figure \ref{fig:SalomonSignal}. However, we also see that the instantaneous ellipticity reaches sporadically $\simeq \pi/4$ (almost circular polarization) and $\simeq 0^+ $ (almost linear polarization), thus revealing more details about the signal.

\begin{figure}
\centering
	\includegraphics[width=\textwidth]{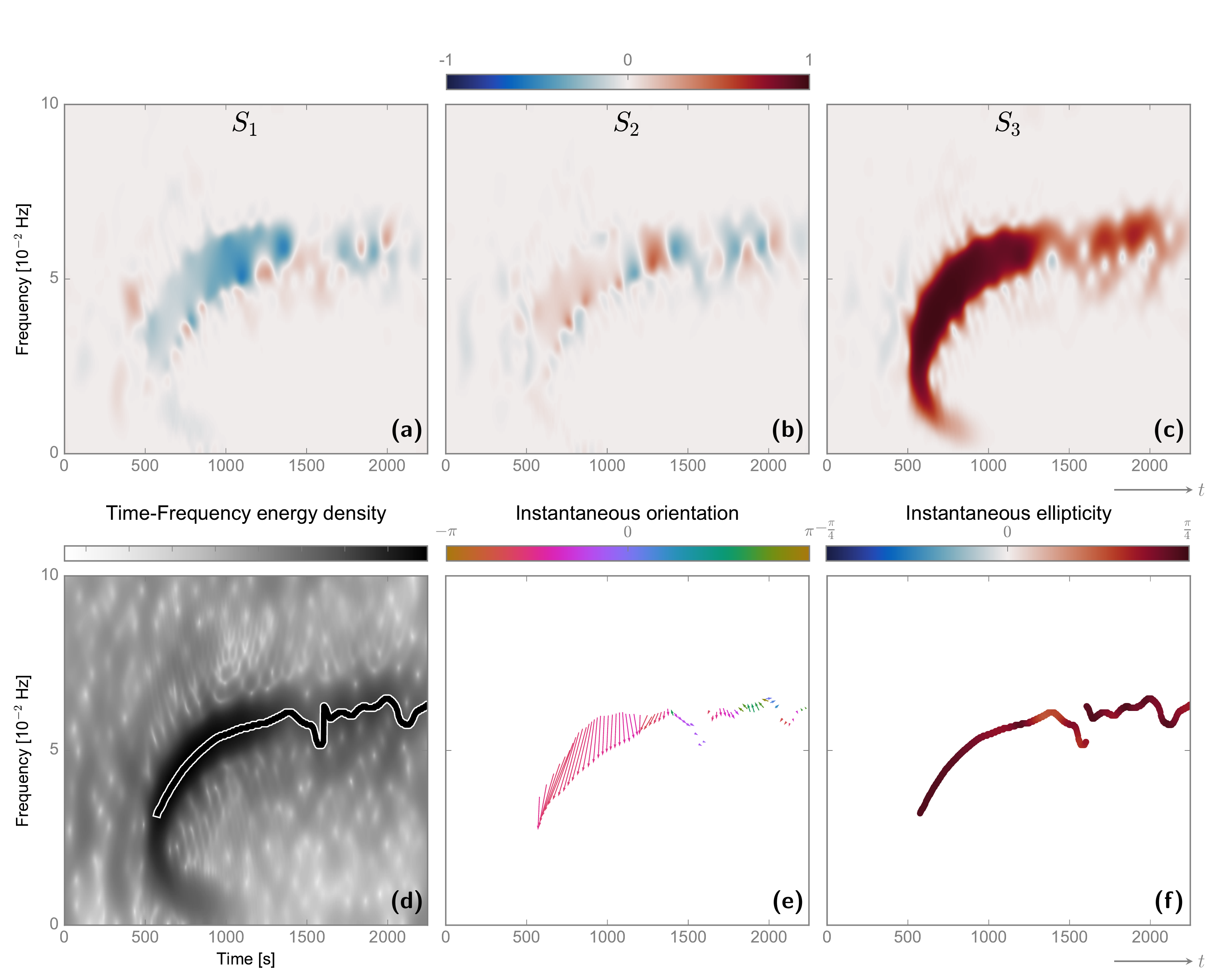}
	\caption{Polarization spectrogram of the Solomon Earthquake dataset (a), (b) and (c) Stokes parameters in the time-frequency plane. (d) Time-frequency energy density (e) instantaneous orientation on the ridge (f) instantaneous ellipticity on the ridge.}\label{fig:SalomonBivariateSpectro}
\end{figure}
\begin{figure}
\centering
	\includegraphics[width=\textwidth]{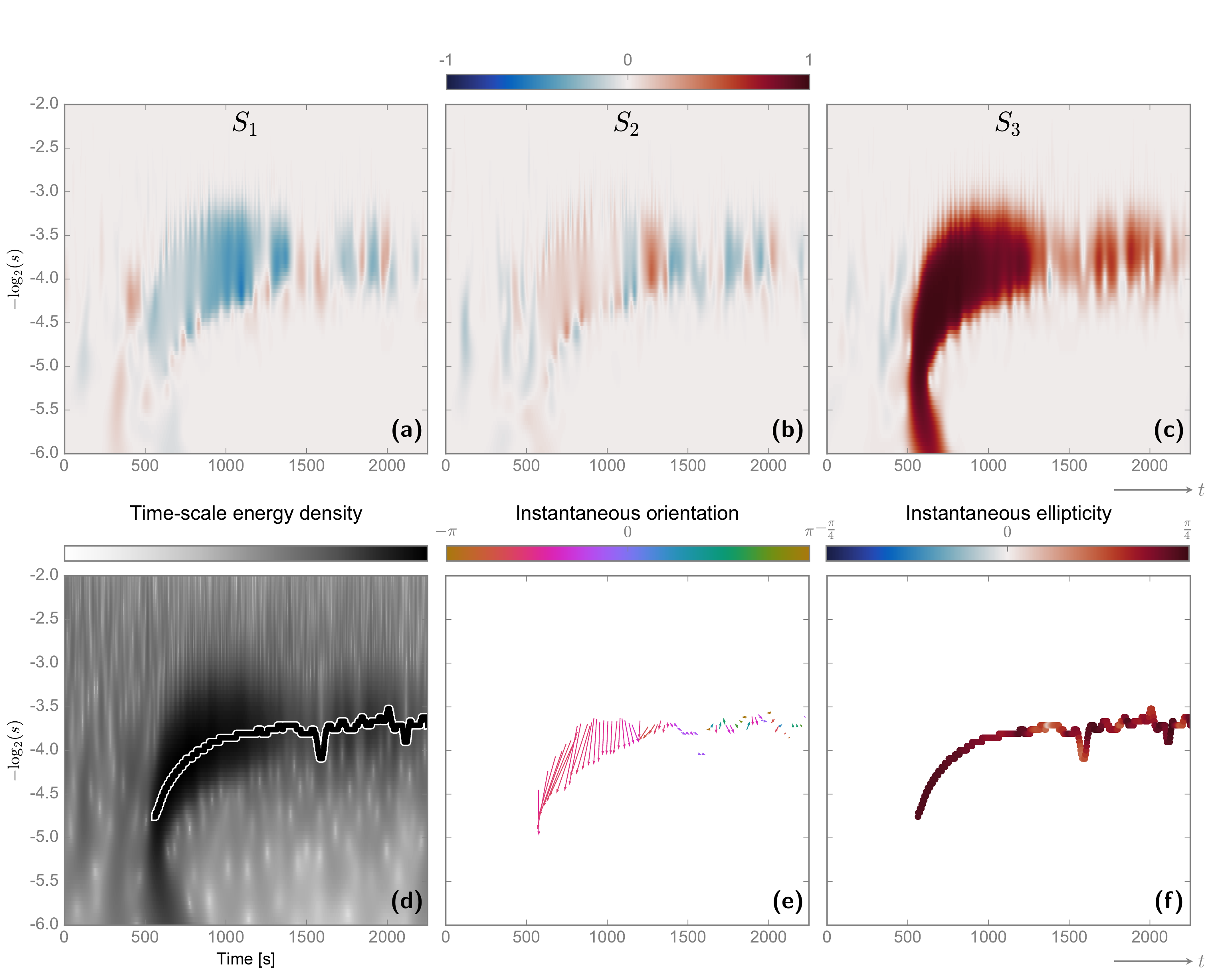}
	\caption{Polarization scalogram of the Solomon Earthquake dataset ((a), (b) and (c) Stokes parameters in the time-scale plane. (d) Time-scale energy density (e) instantaneous orientation on the ridge (f) instantaneous ellipticity on the ridge.}\label{fig:SalomonBivariateScalo}
\end{figure}

\section{Conclusion}

% Conclusion

% Q-FT
We have proposed a generalization of time-frequency analysis to bivariate signals. Our approach is based on the use of a
Quaternion Fourier Transform (QFT). It appears that, despite the apparent complexity of the quaternion algebra due to non-commutativity, the natural extension of usual definitions leads to a bivariate time-frequency toolbox with nice properties. This new framework interestingly includes the usual univariate Fourier analysis. Heisenberg's uncertainty principle is again supported by Gabor's theorem. The definition of the quaternion embedding, a counterpart of the analytic signal for bivariate signals, yields a natural elliptic description of the instantaneous polarization state thanks to the Euler polar form. As a consequence we have introduced instantaneous Stokes parameters, the relevant physical quantities to describe the polarization state of polarized waves.

% Q-STFT and Q-CWT
Turning to time-frequency representations, we have defined a quaternion short term Fourier transform (Q-STFT) and a quaternion continuous wavelet transform (Q-CWT) as simple generalizations of the usual definitions by using the QFT in place of the usual Fourier transform. This extension permits to show fundamental theorems on the conservation of energy and polarization quantities as well as reconstruction formulas. These theorems permit the definition of spectrograms and scalograms, including the representation of the evolution of the polarization state in the time-frequency plane. These spectrogram and scalogram possess an underlying RKHS structure.

% In practice
In practice, due to Gabor's theorem, spectrograms and scalograms are never perfectly localized so that one often needs to extract ridges to accurately identify the time-frequency content of the signal. Classical ridge extraction algorithms apply to the present framework.
The application of the proposed toolbox to synthetic as well as real world data has demonstrated the efficiency of the proposed approach. The resulting graphical representations make the time-frequency content of bivariate signals very readable and intelligible. On a practical ground, the numerical implementation remains simple and cheap since it relies on the use of a few fast Fourier transforms. The code will be made available from our websites.

% Final
We emphasize the general relevance and efficiency of the proposed approach to analyze a wide class of bivariate signals without any {\em ad hoc} model. This approach is very generic. We believe that it will be useful in many applications where the joint time-frequency analysis of 2 components is required. Moreover this work paves the way to the definition of even more general tools to deal with either bivariate signals in dimension larger than 1 or multivariate signals with values in a D-dimensional space with $D\geq 2$.

\section*{Acknowledgments}

Nicolas Le Bihan's research was supported by the ERA, European Union, through the International Outgoing Fellowship (IOF GeoSToSip 326176) program of the 7th PCRD. This work was partly supported by the CNRS, GDR ISIS, within the SUNSTAR interdisciplinary research program.

\appendix
\section{Proofs}

%%%%%%%%%%%%%%%%%%%%% THEOREM 3
\subsection{Proof of theorem \ref{theorem:CompletenessSTQFT}}\label{appendix:proofCompletenessSTQFT}
The time-frequency-polarization atoms $g_{u,\xi}$ are of the form $g_{u, \xi}(t) = g(t-u)\exp(\bmj\xi t)$, where $g$ is a real symmetric window. Let us rewrite the Q-STFT coefficients $Sf(u, \xi)$ like
	\begin{align}
		Sf(u,\xi) & = \intinf f(t)g(t-u)\exp(-\bm j \xi t)\dt \nonumber \\
		&= \left(\intinf f(t)g(u-t)\exp[\bm j \xi (u-t)]\mathrm{d}t\right)\exp(-\bm j \xi u) \nonumber \\
		&= (f\ast g_{0,\xi})(u)\exp(-\bm j \xi u)
	\end{align}
	where $g_{0, \xi}(t) = g(t)\exp(\bm j \xi t)$. The QFT of this expression yields
	\begin{align}
	\label{eq:QFT_STFT}
		\intinf Sf(u, \xi)\exp(-\bm j \omega u) \mathrm{d}u %&= \intinf (f\ast g_{0,\xi})(u)\exp(-\bm j \xi u)\exp(-\bm j \omega u)\mathrm{d}u\\
		&=\intinf (f\ast g_{0,\xi})(u)\exp(-\bm j (\omega + \xi) u)\mathrm{d}u \nonumber\\
		&= \hat{f}(\omega + \xi)\hat{g}_{0,\xi}(\omega + \xi) \nonumber\\
		&= \hat{f}(\omega + \xi)\hat{g}(\omega)
	\end{align}
	thanks to the convolution property (Prop. \ref{proposition:Convolution}) of the QFT. Parseval's formula with respect to $u$ yields
	\begin{equation}
		\intinf Sf(u,\xi)g(t-u)\mathrm{d}u = \frac{1}{2\pi}\intinf \left[\hat{f}(\omega + \xi)\hat{g}(\omega)\right]e^{\bm j t \omega}\overline{\hat{g}(\omega)}\mathrm{d}\omega
	\end{equation}
	Since $g$ is real, its QFT is $\mathbb{C}_{\bm j}$-valued, it commutes with the exponential kernel, \ie $e^{\bm j t \omega}\overline{\hat{g}(\omega)} = \overline{\hat{g}(\omega)}e^{\bm j t \omega}$. 
	We conclude the proof of the inversion formula (\ref{eq:inversionFormulaSTFT}) according to the following lines by using that $\Vert g \Vert = 1$:
	\begin{align}
		\frac{1}{2\pi}\intinf\intinf Sf(u, \xi)g(t-u)e^{\bm j \xi t}\mathrm{d}u\mathrm{d}\xi  &= \frac{1}{4\pi^2}\intinf\intinf \left[\hat{f}(\omega + \xi)\hat{g}(\omega)\right]\overline{\hat{g}(\omega)}e^{\bm j t (\omega+\xi)}\mathrm{d}\omega\mathrm{d}\xi \nonumber \\
		 &= \frac{1}{4\pi^2}\intinf\intinf \hat{f}(\omega + \xi)\vert \hat{g}(\omega)\vert^2e^{\bm j t (\omega+\xi)}\mathrm{d}\omega\mathrm{d}\xi \nonumber\\
		&= \frac{1}{2\pi}\intinf f(t) \vert \hat{g}(\omega)\vert^2 \mathrm{d}\omega  \nonumber\\
		&= f(t)
	\end{align}

	From (\ref{eq:QFT_STFT}) we know that the QFT of $Sf(u, \xi)$ with respect to $u$ is $\hat{f}(\omega + \xi)\hat{g}(\omega)$. Then using the usual Plancherel's formula in $u$ yields
	\begin{align}
		\frac{1}{2\pi}\iint \vert Sf(u, \xi)\vert^2\mathrm{d}u\mathrm{d}\xi &= \frac{1}{2\pi}\int\frac{1}{2\pi}\int \vert \hat{f}(\omega + \xi)\hat{g}(\omega)\vert^2 \mathrm{d}\omega\mathrm{d}\xi \nonumber\\
		&=\frac{1}{2\pi}\int\frac{1}{2\pi}\int \vert \hat{f}(\omega)\hat{g}(\omega-\xi)\vert^2 \mathrm{d}\omega\mathrm{d}\xi \nonumber\\
		&=\frac{1}{2\pi}\int \vert \hat{f}(\omega) \vert^2 \mathrm{d}\omega \nonumber\\
		&= \Vert f \Vert^2.
	\end{align}
	which concludes the proof of the energy conservation property (\ref{eq:energyConservationQSTFT}). The polarization conservation property (\ref{eq:polarConservationQSTFT}) is proven along the same lines using the second Plancherel formula in (\ref{eq:PlancherelQFT}).
\qed

%%%%%%%%%%%%%%%%%%%%% THEOREM 4
\subsection{Proof of theorem \ref{theorem:CompletenessQCWT}}\label{appendix:proofCompletenessCWT}
	We first prove a preliminary result. The polarization wavelets $\psi$ are $\bbCj$-valued and we use the notation $\psi_s(t) = s^{-1/2}\psi(t/s)$. Let $f_s(u) = Wf(u,s)$ the wavelet coefficients at scale $s$. The QFT with respect to $u$ is
	\begin{align}
		\hat{f}_s(\omega)& = \intinf Wf(u,s)\exp(-\bmj \omega u)\mathrm{d}u \nonumber\\
		& = \intinf\intinf f(t)\overline{\psi_s(t-u)}\exp(-\bmj \omega u)\mathrm{d}u\dt \nonumber\\
		& = \intinf f(t)\left(\intinf\overline{\psi_s(t-u)}\exp(-\bmj \omega u)\mathrm{d}u \right)\dt \nonumber\\
		%& = \intinf f(t)\overline{\hat{\psi}_s(\omega)}\exp(-\bmj\omega t)\dt \nonumber\\
		& = \left(\intinf f(t)\exp(-\bmj\omega t)\dt)\right)\overline{\hat{\psi}_s(\omega)} \nonumber\\
		& = \underbrace{\hat{f}(\omega)}_{\in \bbH}\underbrace{\overline{\hat{\psi}_s(\omega)}}_{\in \bbCj}
	\end{align}
	% Now, since $\hat{\psi}(\omega) = 0$ at $\omega < 0$, and that $\hat{f}_+(\omega) = 2\hat{f}(\omega)$ for $\omega \geq 0$, we get that
	% \begin{equation}
	% 	\hat{f_s}(\omega) = \hat{f}(\omega)\overline{\hat{\psi}_s(\omega)} = \frac{1}{2}\hat{f}_+(\omega)\overline{\hat{\psi}_s(\omega)},
	% \end{equation}
	% which is simply the QFT of the first equality in the theorem. Therefore, 
	% \begin{equation}
	% 	Wf(u,s) = \frac{1}{2}Wf_+(u,s).
	% \end{equation}
Let us consider the quantity $b(t)$  
	\begin{equation}
		b(t) = \intinf\intinf Wf(u,s)\psi_s(t-u)\mathrm{d}u\frac{\mathrm{d}s}{s^2}
	\end{equation}
	Taking the QFT of $b(t)$ yields
	\begin{align}
		\hat{b}(\omega) &= \intinf \left(\intinf\int_0^{+\infty} Wf(u,s)\psi_s(t-u)\mathrm{d}u\frac{\mathrm{d}s}{s^2}\right)\exp(-\bmj \omega t)\dt \nonumber\\
		& =\intinf \int_0^{+\infty} Wf(u, s)\left(\intinf \psi_s(t-u)\exp(-\bmj \omega t)\dt\right)\mathrm{d}u\frac{\mathrm{d}s}{s^2} \nonumber\\
		&= \intinf \int_0^{+\infty} Wf(u, s) \hat{\psi}_s(\omega)\exp(-\bmj\omega u)\mathrm{d}u\frac{\mathrm{d}s}{s^2} \nonumber\\
		& =\intinf  \int_0^{+\infty} Wf(u, s)\exp(-\bmj\omega u)\hat{\psi}_s(\omega)\mathrm{d}u\frac{\mathrm{d}s}{s^2} \nonumber \\
		& = \int_0^{+\infty}\hat{f}(\omega)\overline{\hat{\psi}_s(\omega)}\hat{\psi}_s(\omega)\frac{\mathrm{d}s}{s^2}
	\end{align}
	so that 
	\begin{equation}
		\hat{b}(\omega) = \hat{f}(\omega)\intinf \frac{\vert \hat{\psi}(\xi)\vert^2}{\xi}\mathrm{d}\xi =  \hat{f}(\omega) C_\psi,
	\end{equation}
	If the admissibility condition is satisfied (if $C_\psi$ is finite), $f_+(t)$ and $C_\psi^{-1}b(t)$ have the same Fourier transforms. This proves the inversion formula (\ref{eq:inverseReconstructionCWT}).
	\begin{equation}
		f(t) = \frac{1}{C_\psi}\intinf \int_0^{+\infty} Wf(u,s)\psi_s(t-u)\mathrm{d}u\frac{\mathrm{d}s}{s^2}.
	\end{equation}
	%
	%%Energy conservation CWT
	Turning to the energy conservation property (\ref{eq:energyConservationQCWT}), we write according to the first Plancherel's formula  (\ref{eq:PlancherelQFT})
	\begin{equation}
		\frac{1}{C_\psi}\intinf \int_0^{+\infty} \vert Wf_+(u,s)\vert^2\mathrm{d}u\frac{\mathrm{d}s}{s^2} = \frac{1}{C_\psi}\intinf\int_0^{+\infty}\vert\hat{f}_+(\omega)\overline{\hat{\psi}_s(\omega)}\vert^2 \mathrm{d}\omega\frac{\mathrm{d}s}{s^2}
	\end{equation}
	Simplifiyng this expression yields
	\begin{equation}
		\frac{1}{C_\psi}\intinf\int_0^{+\infty}\vert Wf_+(u,s)\vert^2\mathrm{d}u\frac{\mathrm{d}s}{s^2} = \int_{-\infty}^{+\infty}\vert \hat{f}_+(\omega)\vert^2\mathrm{d}\omega = \Vert f_+ \Vert^2
	\end{equation}
	which proves (\ref{eq:energyConservationQCWT}).
	The conservation of polarization properties (\ref{eq:polarConservationQCWT}) is obtained following the same lines.
	\qed

%% Stationay phase approx
\section{Stationary phase approximation}

%% Principle : REVOIR REF ?
\subsection{Principle}\label{appendix:stationaryPhaseApproximation}
We will need a stationary phase approximation to study the localization of ridges in spectrograms and scalograms.
We briefly recall the stationary phase approximation \cite{dingleAsymptotic}. This argument is based on \cite{delprat1992asymptotic}, adapted to the case of the QFT. 
Let us consider the integral
\begin{equation}
	I = \intinf A(t)e^{\bmj \phi(t)}\dt,
\end{equation}
where $A \in \mathcal{C}^\infty_0(\bbR, \bbH)$, which ensures that $\vert A(t)\vert \rightarrow 0$ as $t \rightarrow \pm \infty$, and $\phi \in C^\infty(\bbR, \bbR)$. We assume moreover that the function $\phi$ is varying much faster than variations of $A$. 
%This ensures that the function $\phi$ is oscillating faster than $A$, and that $\vert A(t)\vert \rightarrow 0$ as $t \rightarrow \pm \infty$. 
Let $\tau$ be a stationary point of $\phi$ such that is $\phi'(\tau) = 0$. Assume that $\tau$ is unique, otherwise the contributions of all stationnary points must be summed up. Then, if the conditions of a stationary phase approximation are valid, the change of variable
\begin{equation}
	u(t) = \frac{1}{\sqrt{2}}\left( -\bmj \mbox{sign}(\phi''(\tau)) |\phi''(\tau)| \right)^{1/2} (t-\tau)    
\end{equation}
permits to get the following approximation
\begin{equation}
	I\simeq A(\tau) e^{j\phi(\tau)} \int_{u(\bbR)} e^{-u^2} \left( \frac{du}{dt} \right)^{-1} du.
\end{equation}
Then, thanks to usual results on Gaussian integrals one gets 
\begin{equation}\label{eq:appendixPhasestationaryApproximation}
	I \simeq \sqrt{2\pi}\frac{A(\tau)}{\sqrt{\vert \phi''(\tau)\vert }}e^{\bmj\sign{(\phi''(\tau))}\frac{\pi}{4}}e^{\bmj\phi(\tau)}.
\end{equation}

\subsection{Q-STFT asymptotic analysis}\label{appendix:WindowRidges}
For sake of simplicity, we restrict our analysis to points $(u, \xi) \in \Omega \subset \mathcal{P}$ such that the time-frequency-polarization atoms $g_{u, \xi}$ belong to the Hardy space $H^2(\bbR, \bbCj)$, see (\ref{eq:Hardyspace}). This restriction ensures that $g_{u, \xi}$ features positive frequencies only so that $\forall\: (u, \xi) \in \Omega$,
\begin{equation}\label{eq:HardyspaceConditionSpectro}
	Sf(u, \xi) = \frac{1}{2}Sf_+(u, \xi).
\end{equation}
We consider bivariate AM-FM signals $f(t)$ such that (\ref{eq:polarizedAMFMrequirements}) is valid.
The aim of this section is to show that Q-STFT ridges of these signals concentrate on the local instantaneous frequency $\phi'(u)$ at time $u$. 
We start from
\begin{equation}
	Sf(u, \xi) = \frac{1}{2}\ip{f_+, g_{u, \xi}} = \frac{1}{2}\intinf f_+(t)g(t-u)e^{-\bmj \xi t}\dt.
\end{equation}
Recall that $f_+(t) = a(t)\exp(-\bmk \chi(t))\exp(\bmj \phi(t))$ so that
\begin{equation}\label{eq:Sfexplicit}
	Sf(u, \xi) = \frac{1}{2}\intinf a(t)e^{-\bmk \chi(t)}g(t-u)e^{\bmj (\phi(t)-\xi t)}\dt.
\end{equation}
This expression is an oscillatory integral that can be approximated using a stationary phase argument, see \ref{appendix:stationaryPhaseApproximation} above. Let $\Phi_{\xi}(t) = \phi(t) - \xi t$, and denote $\tau = \tau(\xi)$ a stationary phase point such that $\Phi'_{\xi}(\tau) = 0$. We assume that $\tau(\xi)$ is unique for each $\xi$ and that $\Phi_{\xi}''(\tau) \neq 0$ for simplicity\footnote{If there are multiple stationary points, one must sum their contributions. Also, if $\Phi_{\xi}''(\tau) = 0$, then we search the smallest $k \geq 2$ such that $\Phi_{\xi}^{(k)}(\tau) \neq 0$. Formula follow by straightforward adjustment.}. The stationary phase approximation reads
\begin{equation}
		Sf(u, \xi) \simeq \sqrt{\frac{\pi}{2}}\frac{a(\tau)e^{-\bmk \chi(\tau)}g(\tau-u)}{\sqrt{\vert \phi''(\tau)\vert }}e^{\sign{(\phi''(\tau))}\bmj\frac{\pi}{4}}e^{\bmj(\phi(\tau)-\xi \tau)}
\end{equation}
which can be rewritten as
\begin{equation}\label{eq:Sfasymptotic}
	Sf(u, \xi) \simeq  f_+(\tau)\cdot \sqrt{\frac{\pi}{2}}\frac{g(\tau-u)}{\sqrt{\vert \phi''(\tau)\vert }}e^{\sign{(\phi''(\tau))}\bmj\frac{\pi}{4}}e^{-\bmj\xi \tau}.
\end{equation}
The last equation shows that the approximation involves the computation of the quaternion embedding at stationary points only. This is a straightforward generalization of classical results on the analytic signal \cite{delprat1992asymptotic}.
The \emph{ridge} of the transform is the set of points $(u, \xi) \in \Omega$ such that $\tau(\xi) = u$. On the ridge, one has  
\begin{equation}
	\xi_{\mathrm{ridge}}(u) = \phi'(u),
\end{equation}
which corresponds to the instantaneous frequency of an AM-FM signal. The restriction of the Q-STFT to the ridge then is
\begin{equation}
	Sf(u, \xi_{\mathrm{ridge}}(u)) \simeq  f_+(u)\cdot \sqrt{\frac{\pi}{2}}\frac{g(0)}{\sqrt{\vert \phi''(u)\vert }}e^{\sign{(\phi''(u))}\bmj\frac{\pi}{4}}e^{-\bmj\xi u}.
\end{equation}
%% Q-CWT asymptotic analysis
\subsection{Q-CWT asymptotic analysis}\label{appendix:WaveletRidges}

Let $\psi \in H^2(\bbR, \bbCj)$ be an admissible wavelet, and let $f \in L^2(\bbR, \bbH)$. One has
\begin{equation}
	Wf(u, s) = \frac{1}{2}Wf_+(u, s). 
\end{equation}
In polar form, $\psi(t) = \vert \psi(t)\vert\exp(\bmj \phi_\psi(t))$ and the wavelet transform reads
\begin{equation}
	Wf(u, s) = \frac{1}{2\sqrt{s}}\intinf f_+(t)\overline{\psi\left(\frac{t-u}{s}\right)}\dt.
\end{equation}
Using the Euler polar form of $f_+(t)$ yields
\begin{equation}
	Wf(u, s) = \frac{1}{2\sqrt{s}}\intinf a(t)e^{-\bmk \chi(t)}\left\vert \psi\left(\frac{t-u}{s}\right)\right \vert e^{\bmj (\phi(t) - \phi_{\psi}[(t-u)/s])}\dt.
\end{equation}
This is an oscillatory integral, which is approximated using formula (\ref{eq:appendixPhasestationaryApproximation}). For $(u, s) \in \mathcal{P}$ we assume that $\tau=\tau(u,s)$ is the unique stationary point of $\Phi_{u, s}(t) = \phi(t) - \phi_{\psi}[(t-u)/s]$ such that $\Phi_{u, s}'(\tau) = 0$ and $\Phi''_{u, s}(\tau) \neq 0$. Then
\begin{equation}
	Wf(u, s) \simeq  \sqrt{\frac{\pi}{2s}}f_+(\tau)\frac{1}{\sqrt{\vert \Phi''_{u, s}(\tau)\vert }}\overline{\psi\left(\frac{\tau-u}{s}\right)}e^{\sign(\Phi''_{u, s}(\tau))\bmj\frac{\pi}{4}}
\end{equation}
The ridge of the transform is the set of points $(u, s) \in \mathcal{P}$ such that $\tau(u,s) = u$. On the ridge, one has 
\begin{equation}
	s_{\mathrm{ridge}}(u) = \frac{\phi'_\psi(0)}{\phi'(u)},
\end{equation}
which corresponds to the instantaneous frequency of the signal $f$. The restriction of the wavelet transform to the ridge is:
\begin{equation}\label{eq:QCWTRidgeExpression}
		Wf(u, s_{\mathrm{ridge}}(u)) \simeq  \sqrt{\frac{\pi}{2s_{\mathcal{R}}(u)}}f_+(u)\frac{1}{\sqrt{\vert \Phi''_{u, s_{\mathcal{R}}(u)}(u)\vert }}\overline{\psi(0)}e^{\sign(\Phi''_{u, s_{\mathcal{R}}(u)}(u))\bmj\frac{\pi}{4}}.
\end{equation}
Here we mostly consider wavelet modulated windows $\psi(t) = g(t)\exp(\bmj\eta t)$, where $g$ is a real symmetric window and $\eta >0$ is the central frequency such that $\psi$ belongs to $H^2(\bbR, \bbCj)$. In this case, equation (\ref{eq:QCWTRidgeExpression}) can be rewritten as
\begin{equation}
	Wf(u, s_{\mathrm{ridge}}(u)) \simeq  f_+(u)\cdot \sqrt{\frac{\pi}{2s_{\mathrm{ridge}}(u)}}\frac{g(0)}{\sqrt{\vert \phi''(u)\vert }}e^{\sign(\phi''(u))\bmj\frac{\pi}{4}}.
\end{equation}
%This proves that ridges carry significant information. 

%% Implementation
\section{Efficient implementation of the QFT using complex FFTs}
The discrete QFT can be computed easily using the standard complex FFT algorithm. In the general case where $f$ is $\bbH$-valued, one can write the function $f$ in the following manner
\begin{equation}
	f(t) = f_s(t) + \bm i f_p(t)
\end{equation}
where $f_s$ and $f_p$ are called \emph{simplex} and \emph{perplex} parts, respectively, and are both $\mathbb{C}_{\bm j}$-valued functions. We note that this decomposition highly resemble the Cayley-Dickson decomposition. Similarly we can write the symplectic decomposition of the QFT of $f$, such that
\begin{equation}
	\hat{f}(\omega) = \hat{f}_s(\omega) + \bm i \hat{f}_p(\omega),
\end{equation}
with $\hat{f}_s(\omega)$, $\hat{f}_p(\omega)$ both $\bbCj$-valued. Therefore it is sufficient to compute two FFTs: one on the simplex part and the other one on the perplex part. The procedure for the inverse QFT is similar. Thus the complexity of the discrete QFT algorithm is the same as that of the standard FFT algorithm, that is $\mathcal{O}(N\log N)$. 

\bibliography{ACHA.bib}

\end{document}